\documentclass[letterpaper, 10 pt, conference]{ieeeconf}  

\IEEEoverridecommandlockouts                              

\let\labelindent\relax
\usepackage[inline]{enumitem}   
\makeatletter
\newcommand{\inlineitem}[1][]{%
\ifnum\enit@type=\tw@
    {\descriptionlabel{#1}}
  \hspace{\labelsep}%
\else
  \ifnum\enit@type=\z@
       \refstepcounter{\@listctr}\fi
    \quad\@itemlabel\hspace{\labelsep}%
\fi}
\makeatother

\usepackage{cite}
\usepackage{color}
\usepackage{xcolor}
\usepackage{mathrsfs}

\usepackage{subcaption}
\usepackage{verbatim}
\usepackage{bbold}
\usepackage[draft]{todonotes}   
\usepackage{fancyhdr}
\usepackage{stfloats}
\usepackage{subcaption}
\usepackage{bm}

\usepackage{graphicx}
\usepackage{tabularx}

\usepackage{stfloats}
\usepackage[cmex10]{amsmath}
\usepackage{mathtools, cuted}
\usepackage{array}
\usepackage{algorithmic}
\usepackage{amsfonts}
\usepackage[ruled, vlined]{algorithm2e}
\usepackage{url}
\usepackage{amsthm}
\usepackage{latexsym}
\usepackage{amssymb} 
\usepackage{amsmath} 
\usepackage{subcaption}   
\newcommand{\R}{\mathbb{R}}

\newcommand{\cl}[1]{\mathcal{#1}}
\newtheorem{remark}{Remark}

\newcommand{\Nxy}[2]{\left\|#1 \right\|_{#2}}
\newcommand{\nxy}[2]{\left|#1 \right|_{#2}}

\newcommand{\N}{\mathbb{N}}
\newcommand{\Cxy}[2]{\mathcal{C}(#1,#2)}
\newcommand{\LCxy}[2]{\mathcal{LC}(#1,#2)}
\newcommand{\Csxy}[2]{\mathcal{C}_s(#1,#2)}
\newcommand{\LCsxy}[2]{\mathcal{LC}_s(#1,#2)}
\renewcommand{\b}[1]{\mathbf{#1}}
\newcommand{\CLM}[1]{\b{\Psi}^{\b{#1}}}
\newcommand{\tCLM}[1]{\b{\tilde{\Psi}}^{\b{#1}}}
\newcommand{\tclm}[2]{\tilde{\Psi}^{#1}_{#2}}
\newcommand{\clm}[2]{\Psi^{#1}_{#2}}
\newcommand{\feas}{\mathbf{\overline{\Phi}}_{S_1}}
\newcommand{\TX}{\b{T}^{(x)}}
\newcommand{\TU}{\b{T}^{(u)}}
\newcommand{\Tx}[1]{{T}^{(x)}_{#1}}
\newcommand{\Tu}[1]{T^{(u)}_{#1}}
\newcommand{\WH}{\hat{\b{W}}}
\newcommand{\Wh}[1]{\hat{W}_{#1}}

\newcommand{\ddx}[1]{\frac{\partial}{\partial {#1} }}

\newcommand{\CLMofK}[2]{\b{\Phi}^{\b{#1}}(\b{F},#2)}
\newcommand{\SL}[2]{\mathrm{SL}[#1,#2]}
\theoremstyle{plain}

\newcommand{\lemref}[1]{Lemma \ref{#1}}
\newcommand{\thmref}[1]{Theorem \ref{#1}}
\newcommand{\corref}[1]{Corollary \ref{#1}}
\newcommand{\propref}[1]{Prop. \ref{#1}}

\newcommand{\secref}[1]{Section \ref{#1}}

\newcommand{\figref}[1]{Fig. \ref{#1}}
\newcommand{\defref}[1]{Def. \ref{#1}}

\newcommand{\sat}{\mathrm{sat}}

\theoremstyle{definition}
\newtheorem{coro}{Corollary}[section]
\newtheorem{thm}{Theorem}[section]
\newtheorem{lem}{Lemma}[section] 
\newtheorem{prop}{Proposition}[section]
\newtheorem{defn}[thm]{Definition} 

\newtheorem{rem}[coro]{Remark}

\title{\LARGE \bf
A System Level Approach to Discrete-Time Nonlinear Systems
}

\author{Dimitar Ho 
\thanks{Dimitar Ho is with the Department of Computing and Mathematical Sciences, California Institute of Technology, Pasadena, CA. 
        {\tt\small dho@caltech.edu}}%
}

\begin{document}
\maketitle
\thispagestyle{fancy}
\pagestyle{fancy}


\begin{abstract}
We will show that there is a universal connection between the achievable closed-loop dynamics and the corresponding feedback controller that produces it. This connection shows promise to lead to new methods for robust nonlinear control in discrete-time. We will show that,
given a causal nonlinear discrete-time system and controller, the resulting closed-loop is a solution to a nonlinear operator equation. Conversely, any causal solution to the nonlinear operator equation is a closed-loop that can be achieved by some causal controller. Moreover, solutions can be substituted into a simple dynamic controller structure, which we will refer to as a \textit{system level} controller, to obtain an implementation of the unique corresponding feedback controller. System level controllers could be an attractive approach for robust nonlinear control, as we will show that even when they are parametrized with approximate solutions to the operator equation, they can still produce robustly stable closed loops. We will provide theoretical results that state how grade of approximation and robust stability of the closed loop are related. 
Additionally, we will explore some first applications of our results. Using the cart-pole system as an illustrative example, we will derive how to design robust discrete-time trajectory tracking controllers for continuous-time nonlinear systems. Secondly, we will introduce a particular class of system level controller that shows to be particularly useful for linear systems with actuator saturation and state constraints; The special structure of the controller allows for simple stability and performance analysis of the closed-loop in presence of disturbances. The structure also offers simple ways to do anti-windup compensation, and provides a new nonlinear approach to the constrained LQG problem. A particular application to large-scale systems with actuator saturation and safety constraints is presented in our companion paper \cite{YuHo2020}. 
\end{abstract}

\section{Introduction}
\noindent Compared to linear control theory, there are fewer mathematical tools for tackling controller synthesis of general nonlinear systems. Nevertheless, with the recent explosion of available computational resources and progress in the optimization and control theory community, significant progress has been made towards achieving a more generalized, data-driven approach to nonlinear control design. With the sum-of-squares methods (SOS) \cite{SOSTOOLS}, \cite{papachristodoulou2005tutorial}, it became possible to compute Lyapunov functions for stability analysis through convex optimization. SOS-based controller synthesis methods are presented in \cite{SOSSynthesisPrajnaCT, PrajnaRantzerDensity2004} and \cite{SOSNonlinearWangDT} for the continuous-time (CT) and discrete-time (DT) settings, respectively. Examples of computational methods based on approximating solutions of Hamilton-Jacobi-Bellman type of equations are found in \cite{YPLeong2016},\cite{MatanyaDoyle2014},\cite{JohnHJB}. Other, more recent works \cite{OMLasserre2008} (CT), \cite{HanTedrake2018OccupationMeasures} (DT) provide alternative formulations of optimal controller synthesis through occupation measures.

Inspired by the recently developed system level approach to linear control theory \cite{slsacc}, we will present a new insight on nonlinear discrete-time systems that we believe could lead to entirely new synthesis methods for nonlinear discrete-time systems. The system level approach, as introduced in \cite{slsacc}, enabled new efficient controller synthesis methods \cite{ho2019scalable,SLSWang2018} that allow for localized, distributed and scalable control design in large-scale systems. This is achieved by transforming constrained optimal control problems as convex optimization problems over achievable closed-loop maps that can be solved efficiently. A key component of the system level synthesis (SLS) procedure is that once we have solved for the desired closed-loop map, there is a simple way to construct a controller that stably realizes this on the system.

In this paper we will show that this connection between closed-loop maps and their corresponding realizing controller is not a mere phenomenon of linear systems, but rather a surprisingly universal control principle that extends to general nonlinear discrete time systems. We will show that given a feasible nonlinear closed-loop map from disturbance to state and input, we can follow a procedure to construct an internally stable dynamic controller that realizes the given closed-loop maps. 
More specifically, we will characterize the space of all feasible closed-loop maps as solutions to a nonlinear operator equation and define a dynamic controller that realizes them. In particular, the realizing dynamic controller is obtained by parameterizing a simple controller structure with the solutions of the operator equation. Yet as it turns out, this controller structure, which we will refer to as system level (SL) controller, offers more benefits than its intended original purpose. In fact, we will show that we can parameterize an SL controller with approximate solutions of the operator equation and still obtain stabilizing feedback controllers, if the approximation error is small enough. 
To this end, we will discuss a simple sufficient stability condition based on the small-gain theorem.

The presented approach motivates new paths towards nonlinear control synthesis: 1, finding approximate solutions to the closed-loop operator equation and 2, Obtaining a stabilizing controller by parameterizing an SL controller with the approximate solutions. We will conclude this work by exploring two direct applications of this approach:
\begin{enumerate}
\item We apply the approach to the problem of trajectory tracking for nonlinear continuous-time systems through discrete-time zero-order hold feedback control. As a case study, we evaluate the SL controller on the cart-pole system and demonstrate that despite using only a rough model for synthesis, the resulting controller produces very robust closed loop performance. 
\item We will show that the presented framework gives a systematic way to "blend" multiple linear controllers into one stabilizing nonlinear controller. The resulting structure of the controller fits particularly well into the problem setting of linear systems with actuator saturation and provides new ways to do simple stability and performance analysis of the closed loop. In our companion paper \cite{YuHo2020}, we will show how this technique can be used to improve performance in large-scale linear systems that are subjected to actuator saturation, while guaranteeing that specified safety constraints are never violated. 
\end{enumerate}
We will begin with some notational conventions and mathematical preliminaries.

\section{Preliminaries and Notation}\label{sec:prelim}
\noindent We will define $\ell$ to be the space of real scalar sequences and $\ell^n$ to be the space of all sequences in $\R^n$ over the index set $\N$. Furthermore, sequences will be denoted by small bold letters $\b{x}:=(x_k)^{\infty}_{k=0}$ and occasionally we will define sequences explicitly with the tuple notation $\b{x} := (x_0,x_1,x_2,\dots) $. In addition, we will use the notation $x_{i:j}$ to refer to the truncation of the sequence to the tuple $(x_{i},x_{i+1},\dots, x_{j})$ for ($i <j$) and in reordered form $(x_{i},x_{i-1},\dots, x_{j})$ for ($i >j$).  Denote $\bm{\delta} = (1,0,0,\dots)$ as the scalar unit impulse sequence.\\
\noindent Furthermore, if not otherwise specified, we will use $| \cdot |$ to denote an arbitrary norm in $\R^{n}$ and for matrices $A$, $|A|$ refers to the corresponding induced norm of $A$. $\| \cdot \|$ will be reserved for norms on sequence spaces $\ell$ and $\ell^n$.
We will use the following definition of the norm $\|\cdot\|_{p}$ over vector sequences $\b{x} \in \ell^n$: 
\begin{align*}
&\|\b{x}\|_p := \left(\sum^{\infty}_{k=0}|x_k|^p\right)^{1/p} &\|\b{x}\|_{\infty} := \sup_{k\geq 0} |x_k|
\end{align*}
Correspondingly, define the vector sequence space $\ell^{n}_{p} \subset \ell^n$ as $\ell^{n}_{p} = \{\b{x} \in \ell^n| \|\b{x}\|_p < \infty \}$.

\subsection{Causal Operators}
\noindent Operators will be denoted in bold capital letters $\b{A}$ and will represent maps between vector sequence spaces $\ell^n \rightarrow \ell^m$. An operator will be called causal, if for any pair of input $\b{x} \in \ell^n$ and corresponding output $\b{y} = \b{A}(\b{x})$, the values of $y_t$ do not depend on future input values $x_{t+k}$, $k\geq 1$. More precisely, we will define $\b{A}:\ell^n \rightarrow \ell^m$ to be a \textit{causal} operator if there are functions, $A_t:\R^{n \times (t+1)} \rightarrow \R^m$ that allow $\b{A}$ to be equivalently represented as
\begin{align}\label{def:opdef}
\b{A}(\b{x}) = (A_0(x_0),A_1(x_1,x_0),\dots, A_t(x_{t:0}),\dots)
\end{align}
If in addition, the functions $A_t$ satisfy $A_t(x_{t:0}) = A_t(0,x_{t-1:0})$, i.e. are constant in their first parameter, then $\b{A}$ will be called \textit{strictly causal}.
The functions $\{A_t\}$ fully characterize a causal operator and will also be called \textit{component functions} of $\b{A}$. Notice that every component function $A_t$ has $t+1$ arguments which are populated in reverse-chronological order in Definition \eqref{def:opdef}. For notational convenience define $A_{i:j}: \R^{n \times (\max\{i,j\}+1)} \rightarrow \R^{k \times (|j-i|+1)}$ for $j\geq i$ as
\begin{align*}
A_{i:j}(x_{j:0}) &:= (A_i(x_{i:0}),A_{i+1}(x_{i+1:0}),\dots, A_j(x_{j:0}))
\end{align*}
and for $j<i$ as
\begin{align*}
A_{i:j}(x_{i:0}) &:= (A_i(x_{i:0}),A_{i-1}(x_{i-1:0}),\dots, A_j(x_{j:0})). 
\end{align*}
Define the space of all causal and strictly causal operators $\cl{X} \mapsto \cl{Y}$ as $\Cxy{\cl{X}}{\cl{Y}}$ and $\Csxy{\cl{X}}{\cl{Y}}$, respectively. Similarly, define $\LCxy{\cl{X}}{\cl{Y}}\subset \Cxy{\cl{X}}{\cl{Y}}$ and $\LCsxy{\cl{X}}{\cl{Y}}\subset \Csxy{\cl{X}}{\cl{Y}}$ the space of all linear causal and strictly causal operators.

\subsection{Addition and Multiplication of Operators}
\noindent Sums and products of operators are defined as binary operations on the space of causal operators
where
\begin{align*}
 \b{A}+\b{B}&:\b{x} \mapsto \b{A}(\b{x}) + \b{B}(\b{x})\\
 \b{A}\b{B}\text{ or }\b{A}(\b{B})&: \b{x} \mapsto \b{A}(\b{B}(\b{x})).
 \end{align*}
It is crucial to remember that for general operators, the above defined multiplication is not commutative and is \textbf{only} left-distributed over the summation but \textbf{not} right-distributed, i.e.:
\begin{align*}
 (\b{A}+\b{B})\b{C} = \b{A}\b{C}+\b{B}\b{C}\text{ \underline{but} }\b{C}(\b{A}+\b{B}) \neq \b{C}\b{A}+\b{C}\b{B}
\end{align*}
Moreover, for two operators $\b{A}\in \Cxy{\ell^n}{\ell^m}$, $\b{B}\in \Cxy{\ell^n}{\ell^h}$ with matching domain, $(\b{A},\b{B}) \in\Cxy{\ell^n}{\ell^m \times \ell^h}$ will refer to the operator $(\b{A},\b{B}):\b{x} \mapsto (\b{A}(\b{x}),\b{B}(\b{x}))$ and with slight abuse of notation, we will also use the notation $\b{A} = (\b{A}^1,\b{A}^2)$ to define $\b{A}^1 \in\Cxy{\ell^n}{\ell^m}$ and $\b{A}^2\in\Cxy{\ell^n}{\ell^h}$ as the partial maps of $\b{A}\in\Cxy{\ell^n}{\ell^m \times \ell^h}$. 
\section{closed-loop Maps and Realizing Controllers}
\noindent This section will focus on introducing the notion of closed-loop maps as causal operators w.r.t. general discrete-time nonlinear systems with additive disturbances. Moreover, we will derive necessary and sufficient conditions for operators to be closed-loop maps and how they can be realized by a dynamic controller.\\

\noindent Define the discrete-time nonlinear closed-loop $S_1$ as the dynamical system described by the following equations:
\begin{subequations}
\label{eq:syscl-0}
\begin{align} \label{eq:sys-0}
S_1:&& x_{t} &= f_t(x_{t-1:0},u_{t-1:0}) + w_t, \quad x_0 = w_0\\
\label{eq:u}&& u_t &= K_t(x_{t:0})
\end{align}
\end{subequations}
with $x_t\in \mathbb{R}^n$, $u_t\in \mathbb{R}^m$, $w_t\in \mathbb{R}^n$ being the state, input and disturbance of the system. The functions $f_t:\R^{n \times t} \times \R^{m \times t} \rightarrow \R^n$ characterize the open-loop system behavior and $K_t:\R^{n \times t+1}  \rightarrow \R^m$ represents some causal feedback control applied to the system. Alternatively, we can express the dynamics of $S_1$ as a nonlinear equation in signal space. First, group together the functions $f_t$ and $K_t$ into the strictly causal operator $\b{F}: \ell^n \times \ell^m \rightarrow \ell^n $ and causal operator $\b{K}:\ell^n \rightarrow \ell^m$ as
\begin{align}
\label{eq:defF}\b{F}(\b{x},\b{u}) &:= (0, f_1(x_0,u_0),\dots, f_t(x_{t-1:0},u_{t-1:0}),\dots ) \\
\b{K}(\b{x}) &:= (K_0(x_0), \dots, K_t(x_{t:0}),\dots ).
\end{align} 
Now, we can equivalently define the closed-loop $S_1$ as all trajectories $\b{x}\in\ell^n$, $\b{u}\in\ell^m$ and $\b{w}\in\ell^n$ that are solutions to the nonlinear equation:
\begin{align}\label{eq:syscl-op}
&S_1: & \b{x} = \b{F}(\b{x},\b{u}) + \b{w}, \quad \b{K}(\b{x})=\b{u}
\end{align}
We will use $\b{F}$ and $\b{K}$ to refer to the plant \eqref{eq:sys-0} and controller \eqref{eq:u} of the closed-loop $S_1$ and will refer to \eqref{eq:syscl-op} as the \textit{operator form} of the closed-loop.\\

\noindent From the equations \eqref{eq:sys-0} it is clear that for any disturbance sequence $\b{w}$, the closed-loop $S_1$ produces unique trajectories $\b{x}$ and $\b{u}$. Therefore, the closed-loop $S_1$ induces a well-defined mapping between disturbances $\b{w}$ and state/input trajectories $(\b{x},\b{u})$ and it is clear that the mapping is causal. We will define the map $\b{w} \mapsto (\b{x},\b{u})$ as the closed-loop map $\b{\Phi}_{S_1}[\b{F},\b{K}]$:
\begin{defn}[Closed-Loop Maps]\label{def:clm}
Define $\b{\Phi}_{S_1}[\b{F},\b{K}] \in \Cxy{\ell^n}{\ell^n\times \ell^m}$ as the unique operator that satisfies for all trajectories $\b{w}$, $\b{x}$, $\b{u}$ of $S_1$, the relationship $\b{\Phi}_{S_1}[\b{F},\b{K}](\b{w}) = (\b{x},\b{u})$. $\b{\Phi}_{S_1}[\b{F},\b{K}]$ will be called the closed-loop map of the plant $\b{F}$ w.r.t $\b{K}$. Moreover we will refer to the partial maps $\b{w} \rightarrow \b{x}$ and $\b{w} \rightarrow \b{u}$ with $\b{\Phi}^{\b{x}}_{S_1}[\b{F},\b{K}]$ and $\b{\Phi}^{\b{u}}_{S_1}[\b{F},\b{K}]$, respectively. 
\end{defn}
\noindent Going off of the previous definition, if we leave $\b{K}$ unspecified, we will call $\CLM{} = (\CLM{x},\CLM{u}) \in \Cxy{\ell^n}{\ell^n\times\ell^m}$ a closed-loop map (CLM) of $\b{F}$ if there exists a so-called \textit{realizing} controller $\b{K}'$ such that $\CLM{} = \b{\Phi}_{S_1}[\b{F},\b{K'}]$. Moreover, we will define the set of all such maps $\CLM{}$, the space of all realizable CLMs $\feas[\b{F}]$:
\begin{defn}[Space of Realizable CLMs]
Given a plant $\b{F}$, the space of all realizable closed-loop maps $\feas[\b{F}] \subset \Cxy{\ell^n}{\ell^n \times \ell^m}$ is defined as:
\begin{align*}
\feas[\b{F}]:= \{\b{\Psi}| \exists \text{ causal }\b{K}\text{ s.t. }\b{\Psi}= \b{\Phi}_{S_1}[\b{F},\b{K}] \}
\end{align*}
\end{defn}
\noindent A main result of this paper is the following characterization of the space $\feas[\b{F}]$:
\begin{thm}\label{thm:sufnec}
 $\CLM{} = (\CLM{x}, \CLM{u}) \in \feas[\b{F}]$ if and only if $\CLM{}$ satisfies the operator equation
\begin{align}\label{eq:opfunceq_red}
\CLM{x} = \b{F}(\CLM{}) + \b{I}
\end{align}
Moreover $\b{K'} = \CLM{u}(\CLM{x})^{-1}$ is the unique realizing controller.
\end{thm}
\begin{proof}
See appendix.
\end{proof}
\noindent As it turns out, the space of realizable CLMs $\feas[\b{F}]$ can be precisely characterized as solutions to the nonlinear operator equation \eqref{eq:opfunceq_red}. We will therefore also refer to \eqref{eq:opfunceq_red} as the \textit{CLM equation}.
Writing out the CLM equation \eqref{eq:opfunceq_red} in terms of component functions gives the more explicit condition on the functions $\clm{x}{t}$, $\clm{u}{t}$: The map $\CLM{}=(\CLM{x},\CLM{u})$ satisfies \eqref{eq:opfunceq_red} if and only if its component functions satisfy the following infinite set of function equations for all inputs $w_{t:0}$:
\begin{align}\label{eq:opfunceq_red_fun}
\clm{x}{t}(w_{t:0}) = f_t(\clm{x}{t-1}(w_{t-1:0}),\clm{u}{t-1}(w_{t-1:0})) + w_t
\end{align}
Moreover, \thmref{thm:sufnec} states that the mapping between CLMs $(\CLM{x}, \CLM{u})\in \feas[\b{F}]$ and the corresponding realizing controllers $\b{K}'$ is one-to-one, via the relationship $\b{K'} = \CLM{u}(\CLM{x})^{-1}$. A crucial step of establishing \thmref{thm:sufnec}, is to show that $(\CLM{x})^{-1}$ always exists. This partial result follows from the following important fact, which will be used frequently:
 \begin{prop}\label{prop:invab}
 If $(\b{A}-\b{I}) \in \Csxy{\ell^n}{\ell^n}$ then $\b{A}^{-1} \in \Cxy{\ell^n}{\ell^n}$ exists and $\b{b} = \b{A}^{-1}(\b{a})$ satisfies
$$ b_t = a_t-A_t(0,b_{t-1:0})$$
 \end{prop} 
 \begin{proof}
 Assume given $\b{a}$, we want to find $\b{b}$ s.t. $\b{A}(\b{b}) = \b{a}$. Equivalently, write 
\begin{align}\label{eq:invab}
\b{b} = \b{a} - (\b{A}-\b{I})(\b{b}).
\end{align}
Now, since $\b{A} -\b{I}$ is strictly causal, the component function $A_{t}$ satisfies $A_{t}(x_t,x_{t-1:0}) = A_t(0,x_{t-1:0}) + x_t$. Using this factorization, the component form of \eqref{eq:invab} becomes
\begin{align}\label{eq:invab-comp}
b_t = a_t-A_{t}(0,b_{t-1:0})
\end{align}
and proves existence and uniqueness of $\b{b}$ as it describes a concrete recursive procedure of its computation.
 \end{proof}
  The above proposition explains the guaranteed existence of $(\CLM{x})^{-1}$ in \thmref{thm:sufnec}: The partial map $\CLM{x}$ of a CLM $\CLM{}$ satisfies $\CLM{x} -\b{I}=\b{F}(\CLM{x},\CLM{u})$ due to \eqref{eq:opfunceq_red}. Since $\b{F} \in \Csxy{\ell^n\times \ell^m}{\ell^n}$, we know $\CLM{x} -\b{I} \in \Csxy{\ell^n}{\ell^n}$ and \propref{prop:invab} applies.
 
\subsection{System Level Implementation of Realizing Controllers} 
\noindent Furthermore, the previous \propref{prop:invab} shows us a concrete way to implement the realizing controller $\b{K}'=\CLM{u}(\CLM{x})^{-1}$ of \thmref{thm:sufnec}: Given an input $\b{a}$, the output $\b{b} = \b{K}'(\b{a})$ can be computed recursively through the equations
\begin{subequations}
\label{eq:SL-1}
\begin{align}
c_t &= a_t - \clm{x}{t}(0,c_{t-1:0})\\
b_t &= \clm{u}{t}(c_{t:0}).
\end{align}
\end{subequations}
The above implementation represents a dynamical system with input $\b{a}$, output $\b{b}$ and internal state $\b{c}$ and will be referred to as the \textit{system level implementation} of $\b{K}'$. Moreover, in later sections we will make use of this implementation to define controllers $\b{K}$ that are parametrized by operators $\CLM{x}$, $\CLM{u}$ that are not necessarily CLMs. In particular, the next section will show that such an implementation can give closed-loop stability if $\CLM{}$ is satisfying \eqref{eq:opfunceq_red} approximately. Therefore we will define such structured controllers separately as \textit{System Level} (SL)-controllers:

\begin{defn}\label{def:SL}
Assume given operators $\b{A}\in\Cxy{\ell^n}{\ell^n}$, $\b{B} \in\Cxy{\ell^n}{\ell^m}$, where $\b{A}-\b{I} \in \Csxy{\ell^n}{\ell^n}$.
Consider the dynamical system with input $\b{a}$, output $\b{b}$ and internal state $\b{c}$ according to the equations
\begin{align}
c_t &= a_t - A_t(0,c_{t-1:0})\\
b_t &= B_t(c_{t:0}).
\end{align}
The above dynamical system will be referred to as the system level controller $\SL{\b{A}}{\b{B}}$. 
\end{defn}

\noindent A simple yet important consequence of the definition \defref{def:SL} is that both input $\b{a}$ and output $\b{b}$ can always be expressed through the internal state $\b{c}$ as $\b{a} = \b{A}\b{c}$, $\b{b} = \b{B}\b{c}$.
\section{Robust Stability of closed-loop}\label{sec:robstab}
\noindent As shown in the previous section, any closed-loop map $\CLM{}=(\CLM{x},\CLM{u}) \in \feas[\b{F}]$ can be realized with the corresponding system level controller $\SL{\CLM{x}}{\CLM{u}}$ as defined in \defref{def:SL}. Thus, we know that if we choose $\b{K} = \SL{\CLM{x}}{\CLM{u}}$, then for any disturbance $\b{w}$, the trajectories $(\b{x},\b{u})$ of the closed-loop $S_1$ will be $(\b{x},\b{u}) = \CLM{}(\b{w})$. In this section, we will show that under mild assumptions, the controller $\b{K} = \SL{\CLM{x}}{\CLM{u}}$ guarantees internal stability of the closed-loop. Moreover, we will show that closed-loop stability is guaranteed even if $\CLM{}$ is satisfying the CLM equation \eqref{eq:opfunceq_red} only approximately. 

\noindent To setup the stability analysis we will take the original closed-loop \eqref{eq:sys-0} with $\b{K}$ chosen to be $\SL{\CLM{x}}{\CLM{u}}$ and add additional perturbation signals $\b{v}$ and $\b{d}$ to the internal state of the system level controller and control input. We will call the new perturbed closed-loop $S'_1$:
\begin{subequations}
\label{eq:syscl-2}
\begin{align} 
S'_1:&& x_{t} &= f_t(x_{t-1:0},u_{t-1:0}) + w_t, \quad x_0 = w_0\\
&& \hat{w}_t &= x_t + v_t - \clm{x}{t}(0,\hat{w}_{t-1:0}) \\
&& u_t &= \clm{u}{t}(\hat{w}_{t:0}) + d_t.
\end{align}
\end{subequations}
As before, $w$, $x$ and $u$ represent system disturbance, state and input, and the added state $\hat{w}$ represents the internal state of the system level controller. Furthermore, in line with the definition \defref{def:SL} we will assume $\CLM{x}-\b{I} \in \Csxy{\ell^n}{\ell^n}$, yet we will for now not restrict the maps $\CLM{}$ to lie in the space of feasible CLMs $\feas[\b{F}]$. With $\b{F}$ defined as in \eqref{eq:defF}, $S'_1$ can be equivalently written in operator form as
\begin{subequations}
\label{eq:syscl-2-op}
\begin{align} 
\label{eq:syscl-2-op-x}S'_1:&& \b{x} &= \b{F}(\b{x},\b{u}) + \b{w} \\
\label{eq:syscl-2-op-w}&& \b{\hat{w}} &= \b{x} + \b{v} - (\CLM{x}-\b{I})(\b{\hat{w}}) \\
\label{eq:syscl-2-op-u}&& \b{u} &= \CLM{u}(\b{\hat{w}}) + \b{d}
\end{align}
\end{subequations}
 Analogously to \defref{def:clm}, define $$\b{\Phi}_{S'_1}[\b{F},\CLM{}]: \ell^n \times \ell^m \times \ell^n \mapsto \ell^n \times \ell^m \times \ell^n$$ to refer to the causal mapping $\bm{\delta}:= (\b{w},\b{d},\b{v}) \mapsto (\b{x},\b{u}, \b{\hat{w}})$ between the perturbation signals $\bm{\delta}$ and system/controller states of the closed-loop $S'_1$. Accordingly, define the partial maps $\b{\Phi}^x_{S'_1}[\b{F},\CLM{}]$, $\b{\Phi}^{\hat{w}}_{S'_1}[\b{F},\CLM{}]$ and $\b{\Phi}^{u}_{S'_1}[\b{F},\CLM{}]$.
 
We will analyze stability of the closed-loop $S'_1$ by defining a nominal perturbation $\bm{\delta}^*$ with corresponding nominal response $(\b{x}^*,\b{u}^*, \b{\hat{w}}^*)=\b{\Phi}_{S'_1}(\bm{\delta}^*)$ of $S'_1$ and then investigating how much the closed-loop response $(\b{x},\b{u}, \b{\hat{w}})=\b{\Phi}_{S'_1}(\bm{\delta})$ changes if $\bm{\delta}$ deviates from $\bm{\delta}^*$. We will then call $S_1$ to be $\ell_p$-stable at $\bm{\delta}^*$, if $\|\b{\Phi}_{S'_1}(\bm{\delta})-\b{\Phi}_{S'_1}(\bm{\delta}^*)\|_p$ is small for perturbations $\bm{\delta}$ close to $\bm{\delta}^*$ in the $\ell_p$-sense. More formally, we will use the following notion of $\ell_p$-stability for operators in the statement of our results:

 \begin{defn}\label{def:stablp-1}
An operator $\b{A}\in \Cxy{\ell^n}{\ell^m}$ is called:
\begin{itemize}
 \item $\ell_p$-stable if $\b{A}(\b{a}) \in \ell^m_p$ for all $\b{a} \in \ell^n_p$
  \item \textit{finite gain} (f.g.) $\ell_p$-stable\footnote{If we say $\b{A}$ is f.g. $\ell_p$-stable without specifiying $\b{a}_0$, it is assumed that $\b{a}_0$ is to be taken as $\b{0}$.} at $\b{a}_0 \in \ell^n_p$, if there exists $\gamma, \beta \geq 0$  such that for all $\b{a} \in \ell^n_p$: $$\Nxy{\b{A}(\b{a})-\b{A}(\b{a}_0)}{p} \leq  \gamma\Nxy{\b{a}-\b{a}_0}{p} + \beta .$$
  \item \textit{incrementally finite gain}\footnote{we will write $(\gamma,\beta)$-f.g. and $(\gamma,\beta)$-i.f.g. if we want to specify the constants of the $\ell_p$-stability property} (i.f.g.) $\ell_p$-stable if there exists $\gamma, \beta \geq 0$  such that for all $\b{a},\b{a'} \in \ell^n_p$ holds: $$\Nxy{\b{A}(\b{a})-\b{A}(\b{a'})}{p} \leq  \gamma\Nxy{\b{a}-\b{a'}}{p} + \beta. $$
 \end{itemize}
\end{defn}
\noindent Notice that while the above definition might not be standard, it allows for stability analysis of equilibria, trajectories and limit cycles all within the same definition. 

  \noindent With respect to \defref{def:stablp-1}, the result of \thmref{thm:redstab} presents general conditions for closed-loop stability of $S'_1$: 
    \begin{thm}\label{thm:redstab}
	Assume that $\CLM{}$ and $\b{F}$ is i.f.g $\ell_p$-stable. Then $\b{\Phi}_{S'_1}[\b{F},\CLM{}]$ is f.g. $\ell_p$-stable\footnote{For simplicity assume $\ell_p$-norm over the product space to be chosen as $\|(\b{w},\b{d},\b{v})\|_{p} := \|\b{w}\|_p + \|\b{d}\|_p + \|\b{v}\|_p$.} at $\bm{\delta}^*$  if
	the residual operator $$\b{\Delta}[\b{F},\CLM{}] := \b{F}(\CLM{})+\b{I} - \CLM{x}$$ is $(\gamma,\beta)$-i.f.g. $\ell_p$-stable with $\gamma <1$.
  \end{thm}
  \begin{proof}
  See appendix.
  \end{proof}
\begin{rem}
	Note that $\b{F}$ being i.f.g $\ell_p$-stable does not imply that the open loop system has to be stable. Rather, it can be seen as a generalized smoothness condition. For example: $\b{F}$ is i.f.g $\ell_p$-stable if the components $f_t$ were all chosen as $f_t(x_{t-1:0},u_{t-1:0}) = f(x_{t-1},u_{t-1})$ with some Lipshitz continuous function $f$ in $\R^n$.
\end{rem}  
\begin{rem}
Using the small-gain theorem presented in the Appendix \ref{sec:smallgain-app}, further global and local results can be obtained that do not require the i.f.g property.
\end{rem}
 \noindent  Assuming $\b{F}$ is i.f.g $\ell_p$-stable, an immediate corollary of \thmref{thm:redstab} is that if we choose $\CLM{}\in\feas[\b{F}]$ to be an i.f.g $\ell_p$-stable closed-loop map, then the perturbed closed-loop $S'_1$ is f.g.-$\ell_p$ stable, i.e. this shows that SL controllers implement CLMs in an internally stable way. Moreover, \thmref{thm:redstab} also states that if $\CLM{}$ is an approximate solution to the CLM equation \eqref{eq:opfunceq_red}, then $\SL{\CLM{x}}{\CLM{u}}$ still guarantees robust stability of the perturbed closed-loop $S'_1$. 
 \begin{rem}
 Notice that both \thmref{thm:sufnec} and \thmref{thm:redstab} do not require continuity at any point. Thus, all presented results also apply naturally to the discrete-time hybrid systems setting.
 \end{rem}

\section{Applications for Nonlinear Control Synthesis}

\subsection{Discrete-Time Trajectory Tracking of Continuous-Time Nonlinear System}
\begin{figure*}
\includegraphics[width = \textwidth]{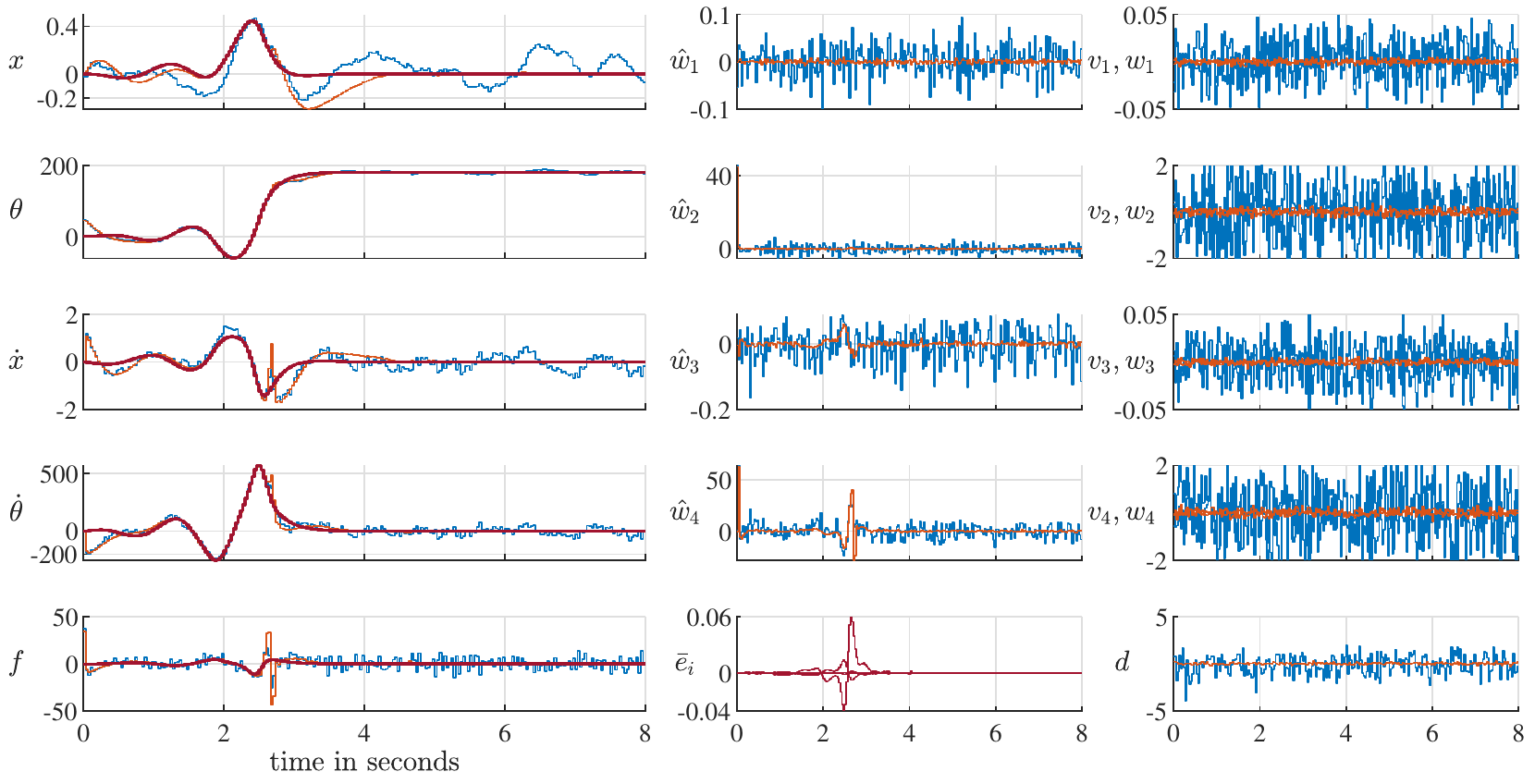}

\caption{Swing-up control for cart-pole system with system level controller at 30Hz sampling rate. ($x$, $\dot{x}$, $\theta$, $\dot{\theta}$, $f$) stand for cart position/velocity, pole angle/velocity and cart force with units ($m$, $m/s$, $deg$, $deg/s$, $N$). $\theta = 180^\circ$ stands for upward pole-position. cart and pole mass is chosen as $1\:$kg and $0.1\:$kg, pole length is chosen as $0.5\:$m. Consult \cite{underactuated} for detailed system description and equations. \textit{Left}: Desired trajectory (red) vs closed-loop performance under initial condition error $\psi(0) = 45^\circ$ and two scenarios of small (orange) and large (blue) system perturbations $\b{w}$, $\b{v}$, $\b{d}$. \textit{Middle}: Evolution of internal state $\hat{w}$ for both scenarios and normalized disturbance due to trajectory error $\bar{e}_i$.  \textit{Right}: i.i.d gaussian perturbations $\b{w}$, $\b{v}$, $\b{d}$. }
\label{fig:invP}
\end{figure*}

\noindent We derived that in theory we only need operators $(\CLM{x},\CLM{u})$ to approximately satisfy the CLM condition \eqref{eq:opfunceq_red} to obtain robustly stabilizing controllers $\SL{\CLM{x}}{\CLM{u}}$. The generality of the robustness result argues that system level controllers could be a promising design tool in practical control applications. Nevertheless, further research needs to quantify the trade-off between grade of approximation of the condition \eqref{eq:opfunceq_red} and corresponding achievable control performance.\\ 
As a first step towards that, we present some first empirical results that show that system level controllers can achieve good robust control performance in challenging-to-control nonlinear systems while using only crude models of the system for synthesis: \figref{fig:invP} shows simulation results of using a system level controller  $\SL{\CLM{x}}{\CLM{u}}$ at $30\,$Hz sampling time to swing up a cart pole system under small and large closed-loop perturbations. Rather than satisfying the CLM condition \eqref{eq:opfunceq_red} of the zero-order hold actuated cart pole system, the maps $\CLM{}$ are synthesized using the following approximations:
\begin{itemize}
\item $\CLM{}$ are taken to be affine operators, where the affine term is a sampled continuous-time desired trajectory $(x^d(t),u^d(t))$ for the system and linear part is chosen to be finite memory ($2\,$s window in continuous-time).
\item The continuous-time trajectories $(x^d(t),u^d(t))$ are low grade approximations of swing-up motions of the cart-pole.
\item $\CLM{}$ are chosen as CLMs of an \textit{approximation} of the linearized system around the desired trajectory. 
\end{itemize} 
\begin{rem}
See Appendix \ref{sec:traj-app} for derivations and more detailed discussion of the synthesis procedure.
\end{rem}
The above simplifications make clear that $\CLM{}$ serves only as a very coarse approximation of the exact CLM condition \eqref{eq:opfunceq_red}. On the other hand, the above approximations allow to synthesize the linear part  of $\CLM{}$ analytically and in parallel, allowing for efficient computation.\\
Leaning on the discussion in \secref{sec:robstab}, the closed-loop $S2$ of the cart pole simulation can be written as
\begin{subequations}
\begin{align}\label{eq:invPcl}
x_{t} &= \phi_{t_s}(x_{t-1},u_{t-1}) + w_t+e_t\\
\hat{w}_t &= x_t+v_t-x^d(t\tau_s)-\sum^{t+1}_{k=2} R_{t,k}\hat{w}_{t+1-k} \\
u_t & = u^d(t\tau_s)+\sum^{t+1}_{k=1} M_{t,k}\hat{w}_{t+1-k}+d_t
\end{align}
\end{subequations}
where $w$, $d$ and $v$ are state, input and internal controller state perturbations and $e_t$ is due to errors in the trajectory synthesis. The matrices $R_{t,k}\in\R^{n\times n}$, $M_{t,k}\in\R^{m\times n}$ are parameterizing the linear part of $\CLM{}$.
Due to the approximation steps taken in the synthesis of $\CLM{}$, there is a considerable gap between the real system and the model used for synthesis. Nevertheless, \figref{fig:invP} shows that despite large model uncertainty, the closed-loop provides robust performance against a variety of perturbations: Large initial condition errors, large perturbation signals and errors in trajectory. 

\subsection{Nonlinear Controller Synthesis for Constrained Linear Systems}
\noindent In our companion paper \cite{YuHo2020}, we show that system level controllers $\SL{\CLM{x}}{\CLM{u}}$ parametrized with a special class of nonlinear maps $\CLM{}=(\CLM{x},\CLM{u})$ are particularly useful for linear systems subject to state/input constraints with actuator saturation. In Appendix \ref{sec:blend-app} we present the generalized procedure to construct such nonlinear system level controllers and derive the robust closed loop stability and performance bounds for linear systems with actuator saturation. 

\noindent The particular controller structure $\SL{\CLM{x}}{\CLM{u}}$ used in \cite{YuHo2020} takes the following general form
\begin{subequations}
\label{eq:nl-blend-large}
\begin{align*}
    \hat{w}_t &= x_t - \sum^{N}_{i=1}\sum^{t+1}_{k=2}R^{i}_{t,k}(P_{\Omega_i}(\hat{w}_{t+1-k})-P_{\Omega_{i-1}}(\hat{w}_{t+1-k}))\\
     u_t &= \sum^{N}_{i=1}\sum^{t+1}_{k=1}M^{i}_{t,k}(P_{\Omega_i}(\hat{w}_{t+1-k})-P_{\Omega_{i-1}}(\hat{w}_{t+1-k}))
\end{align*}
\end{subequations}
where $R^{i}_{t,k}\in \R^{n \times n}$, $M^{i}_{t,k} \in \R^{m \times n}$ and $P_{\Omega_i}(w)$ are projection maps onto some closed bounded convex nested sets $\Omega_1 \subset \dots \subset \Omega_i \subset \dots \subset\Omega_{N-1}$ in $\R^n$. \cite{YuHo2020} shows that the above controllers can be synthesized to merge benefits of different linear controllers together into one nonlinear controller and shows how this proves useful when dealing with optimal control problems that have competing performance and safety objectives. In particular, \cite{YuHo2020} considers the linear optimal control problem of minimizing the closed-loop LQG-cost while also guaranteeing state and input constraints in presence of adversarial disturbances and demonstrates a synthesis procedure for $R^i_{t,k}$, $M^i_{t,k}$ that provably outperforms the optimal linear controller for this problem. An additional benefit of the above SL controller is that it has natural anti-windup properties, i.e. it can be easily synthesized to guarantee closed-loop stability and graceful degradation of performance in presence of actuator saturation (See Appendix \ref{sec:blend-app} for detailed discussion). \\ 
Furthermore, \cite{YuHo2020} points out that the nonlinear "blended" SL controller preserves a main benefit of the original linear SLS framework of \cite{slsacc}: It allows for distributed implementation and naturally incorporates delay and sparsity constraints into the synthesis procedure, all features that allows the approach to scale to the large-scale system setting.

\section{Future Work and Outlook}
The presented discussion and results are in principle applicable or extendable to a broad frontier of possible use cases. Nevertheless, further research is necessary in order to investigate the systems and conditions for which the CLM operator equation can be approximated in a computationally efficient manner. This paper showed some initial results that demonstrate the former can be accomplished with meaningful closed-loop performance in the problem settings of nonlinear trajectory tracking and control synthesis for linear systems with actuator saturation and state constraints. It would be interesting to see if the approach could produce robust trajectory tracking for other nonlinear systems. For example, based on our simulation results, it is conceivable that the approach could be applied to tracking control of trajectories and periodic orbits in other robotic systems.
Another question for future research is whether our discussion in Appendix \ref{sec:blend-app} can be extended to piece-wise linear systems and how other existing approaches in the literature relate to the presented "blended" system level controllers. Furthermore, upcoming work will investigate implications of the presented theory for model predictive control.

\section{Conclusion}
\noindent This paper highlights a surprisingly general and fundamental relationship between closed-loop maps and corresponding realizing controllers in general nonlinear discrete-time systems. The key findings are: 1. All closed-loop maps are solutions to an operator equation and all solutions of the equation are achievable closed-loop maps. 2. Given a solution of the operator equation, we can obtain a realizing controller by parameterizing a system level controller with the solution.  This controller then imposes the given solution as the closed-loop map of the system. 3. This same procedure produces robust closed-loop stability even when the system level controllers are parametrized with approximate solutions of the operator equation.
We conclude with an outlook on the new opportunities this framework could provide for nonlinear control synthesis.
\bibliographystyle{IEEEtran}
\bibliography{IEEEabrv,conbib}

\onecolumn

\appendices \label{app}
\section{Proofs of \thmref{thm:sufnec} and \thmref{thm:redstab}}
\subsection{Proof of \thmref{thm:sufnec}}
\noindent\textit{Necessity:} Per definition of $\feas[\b{F}]$, $\b{x} = \CLM{x}(\b{w})$, $\b{u} = \CLM{u}(\b{w})$ satisfy the dynamics \eqref{eq:syscl-op} and therefore for any $\b{w}$ holds $\CLM{x}(\b{w}) = \b{F}(\CLM{}(\b{w}))+\b{w}$. \textit{Sufficiency}: Assume $\CLM{}$ is a solution of \eqref{eq:syscl-op}, then $\CLM{x}-\b{I} = \b{F}(\CLM{}) \in \Csxy{\ell^n}{\ell^n}$ since $\b{F}\in\Csxy{\ell^n}{\ell^n}$. Applying \propref{prop:invab}, shows existence of $(\CLM{x})^{-1} \in \Cxy{\ell^n}{\ell^n}$. Now take $\b{K}' = \CLM{u}(\CLM{x})^{-1}$ and let $(\b{x},\b{u}) = \b{\Phi}_{S_1}[\b{F},\b{K}'] (\b{w})$, then $\b{x} = \b{F}(\b{x},\CLM{u}(\CLM{x})^{-1}\b{x}) + \b{w}$. We will apply the identity $(\CLM{x}-\b{I})(\CLM{x})^{-1} + (\CLM{x})^{-1} = \b{I}$ to the left side of the equation and obtain
\begin{align*}
(\CLM{x}-\b{I})(\CLM{x})^{-1} + (\CLM{x})^{-1} \b{x} = \b{F}\CLM{}(\CLM{x})^{-1}\b{x} + \b{w}\\
\Leftrightarrow (\CLM{x}-\b{I}- \b{F}\CLM{})(\CLM{x})^{-1} + (\CLM{x})^{-1}\b{x} =  + \b{w} 
\end{align*}
which due to $\CLM{x}-\b{I} = \b{F}(\CLM{})$, implies $\b{x} = \CLM{x}\b{w}$ and $\b{u} = \CLM{u}(\CLM{x})^{-1}\CLM{x}\b{w} = \CLM{u}\b{w}$.\\
Assume $(\CLM{x},\CLM{u}) = \b{\Phi}_{S_1}[\b{F},\b{K}'] = \b{\Phi}_{S_1}[\b{F},\b{K}'']$ for some $\b{K}', \b{K}''$. Then for all $\b{w}$ holds $\CLM{u}(\b{w})=\b{K}'\CLM{x}(\b{w}) = \b{K}''\CLM{x}(\b{w})$ and since $\CLM{x}$ is invertible, we imply $\b{K''} = \b{K'}$.
\subsection{Proof of \thmref{thm:redstab}}
\noindent Writing $\b{\Phi}_{S'_1}(\bm{\delta})$ as $(\CLM{x},\CLM{u},\b{I})(\b{\Phi}^{\hat{w}}_{S'_1}(\bm{\delta})) + (-\b{v}, \b{d}, \b{0})$ and knowing that $\CLM{}$ is i.f.g. tells us that it's sufficient to prove f.g. $\ell_p$-stability of $\b{\Phi}^{\hat{w}}_{S'_1}$:
Take $\b{\hat{w}} = \b{\Phi}^{\hat{w}}_{S'_1}(\bm{\delta})$ and substitute equation \eqref{eq:syscl-2-op-x} into \eqref{eq:syscl-2-op-w} to obtain 
  $\b{\hat{w}} = \b{\Delta}[\b{F},\CLM{}](\b{\hat{w}}) + \bm{\varepsilon}$ where $\bm{\varepsilon}=\b{F}(\CLM{}(\b{\hat{w}})-(\b{v},-\b{d}))-\b{F}(\CLM{}(\b{\hat{w}}))+\b{w} +\b{v}$. Now, since $\b{F}$ is i.f.g. $\ell_p$-stable, we know $\bm{\varepsilon} \in \ell^n_p$. Similarly, $\b{\hat{w}}^* = \b{\Phi}^{\hat{w}}_{S'_1}(\bm{\delta}^*)$ satisfies $\b{\hat{w}}^* = \b{\Delta}[\b{F},\CLM{}](\b{\hat{w}}^*) + \bm{\varepsilon}^*$, with $\bm{\varepsilon}^* \in \ell^n_p$. Taking the difference of the equations for $\b{\hat{w}}$ and $\b{\hat{w}}^*$ and applying \lemref{lem:smg} gives the desired result.

\begin{lem}\label{lem:smg}[Small Gain Theorem]
If $\b{a} = \b{A}(\b{a}) + \b{b}$, $\b{b} \in \ell^n_p$ with $\b{A} \in \Csxy{\ell^n}{\ell^n}$ and $\b{A}$ is $(\gamma,\beta)$ $\ell_p$-stable with $\gamma<1$, then there exist $(\gamma',\beta')<\infty$ such that $\|\b{a}\|_p \leq \gamma'\|\b{b}\|_p + \beta'$.
\end{lem}

\section{A Brief Discussion on Small Gain Theorem}\label{sec:smallgain-app}
\noindent For convenience of the reader, we will present a self-contained discussion of the small-gain theorem used in this paper. Recall \secref{sec:prelim} for nomenclature.
\subsection{Useful Properties of the Truncation Operator}
\begin{defn}\label{def:trunc}[Truncation Operator]
Define $\b{P}^{\tau}\in\Cxy{\ell^n}{\ell^n}$ as $\b{P}^{\tau}\left(\b{x} \right):= (x_0,x_1,\dots, x_{\tau},0,0,\dots )$.
\end{defn}
\noindent The truncation operator is a projection map and satisfies the following identities:
\begin{coro}\label{coro:trunc}
$\b{P}^{\tau}\in\Cxy{\ell^n}{\ell^n}$ as defined in \defref{def:trunc} satisfies the properties:
\begin{enumerate}[label=(\roman*)]
\item $\b{P}^{\tau}(\b{x}) \in \ell^{n}_p$ for all $\tau$ and $\b{x}$
\inlineitem $\b{P}^\tau(\b{x}+\b{y}) = \b{P}^\tau(\b{x}) + \b{P}^\tau(\b{y})$ 
\inlineitem $\b{P}^{\tau}\b{P}^{\tau} = \b{P}^{\tau}$ 
\item $P^{\tau}_{t:0}(\b{x}) = x_{t:0}$ for $t \leq \tau$.
\inlineitem If $\b{x} \in \ell^n_p$, ($p<\infty$) then $\|\b{x}\|^p_p = \|(\b{I}-\b{P}^\tau)(\b{x})\|^p_p+\|\b{P}^\tau(\b{x})\|^p_p$
\item If $\b{x} \in \ell^n_\infty$, then $\|\b{x}\|_\infty = \max\{\|(\b{I}-\b{P}^\tau)(\b{x})\|_\infty,\|\b{P}^\tau(\b{x})\|_\infty\}$
\end{enumerate}
\end{coro}
\noindent With the definition \defref{def:trunc}, we can equivalently characterize causal and strictly causal operators as:
\begin{defn}
An operator $\b{Q}:\ell^n \mapsto \ell^m$ is called causal ($\b{Q}\in\Cxy{\ell^n}{\ell^m}$) if $\b{P}^{\tau}\b{Q} = \b{P}^{\tau}\b{Q}\b{P}^{\tau}$ and strictly causal ($\b{Q}\in\Csxy{\ell^n}{\ell^m}$) if $\b{P}^{\tau}\b{Q} = \b{P}^{\tau}\b{Q}\b{P}^{\tau-1}$. 
\end{defn}
\noindent The following Lemma will be used in the derivation of the Small Gain Theorem:
\begin{lem}\label{lem:caus-trunc}
For any $\b{Q}\in\Cxy{\ell^n}{\ell^m}$ holds $\|\b{P}^\tau\b{Q}(\b{x})\|_p \leq \|\b{Q}\b{P}^\tau(\b{x})\|_p$
\end{lem}
\begin{proof}
Notice that if $\|\b{Q}\b{P}^\tau(\b{x})\|_p = \infty$ the statement is true, since $\|\b{P}^\tau\b{Q}(\b{x})\|_p< \infty$, regardless of $\b{Q}$ (see \corref{coro:trunc}).
\noindent \underline{$p<\infty$:} $\Nxy{\b{Q}\b{P}^{\tau}(\b{x})}{p}^p$ can be decomposed as
\begin{align}
\notag &\Nxy{\b{Q}\b{P}^{\tau}(\b{x})}{p}^p = \sum^{\infty}_{k=\tau+1} \nxy{Q_{k}(0_{n \times k-\tau-1},x_{\tau:0})}{}^p + \sum^{\tau}_{k'=0} \nxy{Q_{k'}(x_{k':0})}{}^p
 = \Nxy{(\b{I}-\b{P}^{\tau})\b{Q}(\b{x})}{p}^p + \Nxy{\b{P}^{\tau}(\b{Q}(\b{x}))}{p}^p
\end{align}
and gives us directly the inequality $\Nxy{\b{P}^{\tau}(\b{Q}(\b{x}))}{p} \leq \sqrt[p]{\Nxy{(\b{I}-\b{P}^{\tau})\b{Q}(\b{x})}{p}^p + \Nxy{\b{P}^{\tau}(\b{Q}(\b{x}))}{p}^p} = \Nxy{\b{Q}(\b{P}^{\tau}(\b{x}))}{p}  $.\\
\noindent \underline{$p=\infty $:} 
$\Nxy{\b{Q}(\b{P}^{\tau}(\b{x}))}{\infty}$ can be decomposed as 
\begin{align}
\notag &\Nxy{\b{Q}(\b{P}^{\tau}(\b{x}))}{\infty}= \max\left\{ \sup_{k \geq \tau+1} \nxy{Q_{k}(0,\dots,x_{\tau:0})}{}, \sup\limits_{k'\leq\tau}\nxy{Q_{k'}(x_{k':0})}{}\right\}
 =  \max\left\{\Nxy{(\b{I}-\b{P}^{\tau})\b{Q}(\b{x})}{\infty}, \Nxy{\b{P}^{\tau}(\b{Q}(\b{x}))}{\infty}\right\}
\end{align}
and similarly leads to $\Nxy{\b{P}^{\tau}(\b{Q}(\b{x}))}{\infty} \leq \Nxy{\b{Q}(\b{P}^{\tau}(\b{x}))}{\infty}$.
\end{proof}

\subsection{A Local Small Gain Theorem}
\noindent Consider the dynamical system in operator form:
\begin{align}\label{eq:ydyn}
\b{y} = \b{\Delta}(\b{y}) + \b{w} \Leftrightarrow (\b{I}-\b{\Delta})(\b{y}) =  \b{w}
\end{align}
where $\b{y}, \b{w} \in \ell^n$ and $\b{\Delta}\in\Csxy{\ell^n}{\ell^n}$. Recall that the Proposition
\begin{prop}\label{prop:invab-2}
 If $(\b{A}-\b{I}) \in \Csxy{\ell^n}{\ell^n}$ then $\b{A}^{-1} \in \Cxy{\ell^n}{\ell^n}$ exists and $\b{b} = \b{A}^{-1}(\b{a})$ satisfies
$ b_t = a_t-A_t(0,b_{t-1:0})$
 \end{prop} 
\noindent tells us that $(\b{I}-\b{\Delta})^{-1}\in\Cxy{\ell^n}{\ell^n}$ exists and we can write the mapping from $\b{w}$ to $\b{y}$ as:
\begin{align}\label{eq:ydyn-2}
\b{y} = (\b{I}-\b{\Delta})^{-1}\b{w}
\end{align}
\begin{lem}\label{lem:smallgain}
Assume that for some $\rho>0$ and $0 \leq \gamma <1$, $\beta\geq0$ the operator $\b{\Delta}$ satisfies $\Nxy{\b{\Delta}(\b{x})}{p} \leq \gamma \Nxy{\b{x}}{p}+\beta$ for all $\Nxy{\b{x}}{p} < \rho$ and some $p \in \{1,2,\dots, \infty\}$. Then, for any $\b{w}$ bounded by $\Nxy{\b{w}}{p} < (1-\gamma)\rho - \beta$, the corresponding $\b{y}$ as defined in \eqref{eq:ydyn-2} satisfies the bound:
\begin{align}\label{eq:smallgain}
\Nxy{\b{y}}{p} \leq \frac{1}{1-\gamma}(\Nxy{\b{w}}{p}+\beta)
\end{align}
\end{lem}
\begin{proof}
%
Pick an arbitrary $\b{w}$ such that $\Nxy{\b{w}}{p} < (1-\gamma)\rho-\beta$ and notice 
\begin{align}\label{eq:redef}
\Nxy{\b{w}}{p} < (1-\gamma)\rho-\beta \Leftrightarrow \frac{\Nxy{\b{w}}{p} +\beta}{(1-\gamma)}<\rho.
\end{align}
Now, let $\b{y}$ be the corresponding response of the dynamic system \eqref{eq:ydyn}. Furthermore, define the scalar sequence $s_{\tau}:=\|\b{P}^{\tau}(\b{y})\|_{p}$, i.e.:
\begin{align}\label{eq:staudef} 
s_{\tau} := \left\{ \begin{array}{cl} \sqrt[p]{\sum^{\tau}_{k=0} \nxy{y_k}{}^p } & \text{for }p < \infty \\ \sup_{k\leq \tau} \nxy{y_k}{}  & \text{for }p = \infty  \end{array} \right..
\end{align}
We will show $s_\tau \leq (\Nxy{\b{w}}{p}+\beta)/(1-\gamma)$ for all $\tau$ per induction:\\

\noindent \underline{$\tau = 0$}: $s_0 = |w_0| \leq \Nxy{\b{w}}{p} \leq (\Nxy{\b{w}}{p}+\beta)/(1-\gamma)$. \\
\noindent \underline{$\tau \rightarrow \tau + 1$:}
Assume $s_{\tau}$ satisfies $s_{\tau} \leq (\Nxy{\b{w}}{p}+\beta)/(1-\gamma)$. Then, due to \eqref{eq:redef}, we have $s_{\tau}=\|\b{P}^{\tau}(\b{y})\|_p < \rho$ and using the small gain property we know: 
\begin{align}\label{eq:smallgain-step}
\Nxy{\b{\Delta}(\b{P}^{\tau}(\b{y}))}{p} \leq \gamma \Nxy{\b{P}^{\tau}(\b{y})}{p}+\beta.
\end{align}
 Then, substituting the dynamics \eqref{eq:ydyn} into the definition \eqref{eq:staudef} and using strict causality of $\b{\Delta}$ gives us
\begin{align}
s_{\tau+1}=\Nxy{\b{P}^{\tau+1}(\b{y})}{p} = \Nxy{\b{P}^{\tau+1}(\b{\Delta}(\b{y}) + \b{w})}{p} =\Nxy{\b{P}^{\tau+1}(\b{\Delta}(\b{P}^{\tau}\b{y}) + \b{w})}{p},
\end{align}
and by using property (iii) of \lemref{lem:caus-trunc} we obtain the following chain of inequalities:
\begin{subequations}
\label{eq:stausteps-1}
\begin{align}
\label{eq:pdelta}s_{\tau+1}=&\Nxy{\b{P}^{\tau+1}\b{\Delta}(\b{P}^{\tau}\b{y}) + \b{P}^{\tau+1}\b{w}}{p} 
\leq \Nxy{\b{P}^{\tau+1}(\b{\Delta}(\b{P}^{\tau}\b{y}))}{p} + \Nxy{\b{w}}{p}
\leq  \Nxy{\b{\Delta}(\b{P}^{\tau}(\b{y}))}{p}+\Nxy{\b{w}}{p}.  
\end{align}
\end{subequations}
Now, we can further upperbound \eqref{eq:pdelta}, by using \eqref{eq:smallgain-step} with our induction assumption $s_{\tau} \leq (\Nxy{\b{w}}{p}+\beta)/(1-\gamma)$:
\begin{subequations}
\label{eq:stausteps-2}
\begin{align}
s_{\tau+1}\leq & \gamma \Nxy{\b{P}^{\tau}(\b{y})}{p}+\beta+\Nxy{\b{w}}{p} = \gamma s_{\tau} + (1-\gamma)\frac{\Nxy{\b{w}}{p}+\beta}{1-\gamma} \\
\label{eq:indcomp}\leq &\gamma \frac{\Nxy{\b{w}}{p}+\beta}{1-\gamma} + (1-\gamma)\frac{\Nxy{\b{w}}{p}+\beta}{1-\gamma} = \frac{\Nxy{\b{w}}{p}+\beta}{1-\gamma},
\end{align}
\end{subequations}
Hence, \eqref{eq:indcomp} completes the induction step and we can conclude that $s_\tau \leq (\Nxy{\b{w}}{p}+\beta)/(1-\gamma) < \rho$ holds for all $\tau$.\\

\noindent Finally, since $s_\tau$ is non-decreasing per construction, we know that $\lim_{\tau \rightarrow \infty} s_{\tau} = s^* = \Nxy{\b{y}}{p}$ exists and satisfies $$\Nxy{\b{y}}{p} \leq \frac{1}{1-\gamma}(\Nxy{\b{w}}{p}+\beta).$$
\end{proof}

\noindent The following Theorems follow as Corollaries of \lemref{lem:smallgain}:
\begin{thm}
Consider system \eqref{eq:ydyn} and assume that $\b{\Delta}$ satisfies $\Nxy{\b{\Delta}(\b{x})}{p} \leq \gamma \Nxy{\b{x}}{p}+\beta$ for all $\b{x} \in \ell^n_p$. Then, correspondingly for all $\b{w} \in \ell^n_p$, the system response $\b{y}$ satisfies the bound  
$$\Nxy{\b{y}}{p} \leq \frac{1}{1-\gamma}(\Nxy{\b{w}}{p}+\beta).$$
\end{thm}
\begin{proof}
Use \lemref{lem:smallgain} with $\rho = \infty$.
\end{proof}

\begin{thm}
Similar to \lemref{lem:smallgain}, assume that for some $\rho>0$ and $0 \leq \gamma <1$, the operator $\b{\Delta}$ satisfies $\Nxy{\b{\Delta}(\b{x})}{p} \leq \gamma \Nxy{\b{x}}{p}$ for all $\Nxy{\b{x}}{p} < \rho$. Then, at $\b{w}=0$, the operator $(\b{I}-\b{\Delta})^{-1}$ is a locally continuous map $\ell^n_p \mapsto \ell^n_p$.
\end{thm}
\begin{proof}
Fix $\varepsilon' > 0$ and pick $\delta' = (1-\gamma)\min\{\rho,\varepsilon'\}$. To apply \lemref{lem:smallgain}, substitute $\rho$ in the setup of \lemref{lem:smallgain} with $\min\{\rho,\varepsilon'\}$, set $\beta=0$ and we conclude that for all $\Nxy{\b{w}}{p} < \delta'$ holds 
\begin{align*}
\Nxy{\b{y}}{p}  =& \Nxy{(\b{I}-\b{\Delta})^{-1}\b{w}}{p} \leq \frac{1}{1-\gamma} \Nxy{\b{w}}{p}\\
 <& \frac{\delta'}{1-\gamma} = \min\{\rho,\varepsilon'\}\leq\varepsilon'
 \end{align*}

\end{proof}

\begin{rem}
If we associate $w_0$ as the initial condition $y_0$ and we can rewrite bounds of \lemref{lem:smallgain} as:
\begin{subequations}
\begin{align}
&\Nxy{\b{y}}{p} \leq \frac{1}{1-\gamma}\left( \sqrt[p]{\nxy{y_0}{p}^p+\Nxy{\b{w}}{p}^p} + \beta\right) &  1\leq p < \infty \\
&\Nxy{\b{y}}{\infty} \leq \frac{1}{1-\gamma}\left(\max\left\{\nxy{y_0}{},\Nxy{\b{w}}{\infty}\right\}+\beta\right) & p=\infty
\end{align}
\end{subequations}
\end{rem}

\section{Discrete-Time Trajectory Tracking Control for Nonlinear Continuous-Time Systems}\label{sec:traj-app}
\noindent Using the cart pole system as an example, we will demonstrate how to develop a system level controller to track trajectories for nonlinear continuous-time systems. We will use the description of the cart pole as presented in \cite{underactuated} and refer to the same reference for detailed derivations. The dynamic equations of the cart pole are
\begin{subequations}
\label{eq:cartpole}
\begin{align}
      (m_c + m_p)\ddot{x}_c + m_p l \ddot\theta_p \cos\theta_p - m_p l \dot\theta^2_p \sin\theta_p = f \\
      m_p l \ddot{x}_c \cos\theta + m_p l^2 \ddot\theta_p + m_p g l \sin\theta_p = 0
\end{align}
\end{subequations}
where $x_c$ and $\theta_p$ stand for cart position and pole angle in counterclockwise direction and $f$ represents the force exerted on the cart. Furthermore, $\theta =0$ denotes the downward position. The parameters ($m_c$, $m_p$, $l$, $g$) are chosen as ($1$ kg, $0.1$ kg, $0.5$ m, $9.81 \mathrm{m}/\mathrm{s}^2$) and represent cart and pole mass, pole length and the gravity constant, respectively. Furthermore, \eqref{eq:cartpole} can be converted into the input affine standard form 
\begin{align}\label{eq:inputaffine}
\dot{x} = F(x)+g(x)u
\end{align}\\
where $x = [x_c,\theta_p,\dot{x}_c,\dot{\theta}_p]^T$, $u = f$, (see \cite{underactuated} for description of $F(x)$ and $g(x)$). 
As in practice controllers are usually implemented digitally, we will assume zero-order hold on the input $u$ with a sampling time of $\tau_s = 0.033 \mathrm{sec}$ ($1/\tau_s= 30\mathrm{Hz}$). Because of this discretization, we can equivalently represent the system \eqref{eq:inputaffine} at sampling times through the discrete-time system 
\begin{align}\label{eq:dtphi}
&x_{t} = \phi_{\tau_s}(x_{t-1},u_{t-1}), \quad \phi_{\tau_s}(x,u):= \alpha(\tau_s),\text{ s.t. : } \dot{\alpha} = F(\alpha,u), \alpha(0) = x,
\end{align}
where we will denote $x_t := x(t \tau_s)$ and $u_t := u(t \tau_s)$ ($t \in \mathbb{N}$) to be samples of the continuous-time signals $x(\tau)$, $u(\tau)$ at time $t \tau_s$. To put \eqref{eq:dtphi} into operator form, define $\b{F}^{\phi}\in \Csxy{\ell^n\times\ell^m}{\ell^n}$ with the component functions $F^\phi_t(x_{t:0},u_{t:0}):= \phi_{\tau_s}(x_{t-1},u_{t-1})$ and equation \eqref{eq:dtphi} can be written in terms of the trajectories $(\b{x},\b{u})$ as $\b{x} = \b{F}^{\phi}(\b{x},\b{u})$.
\begin{rem}
We will use the variable $\tau$ to indicate that a variable $a(\tau)$ is a continuous-time signal and use $a_t$ to refer to the discrete-time samples $a_t := a(t\tau_s)$.
\end{rem}

We will use an approximate trajectory reference $x^d(\tau), u^d(\tau)$ computed from the continuous-time model \eqref{eq:inputaffine} as shown in red in \figref{fig:invP}. The trajectories approximately satisfy the continuous-time system $\dot{x}^d(\tau) \approx F(x^d(\tau)) + g(x^d(\tau)) u^d(\tau))$ and are designed to swing up the pole in 3 seconds and then keep the cart pole at $x_c=0$, $\theta_p = \pi$. 
\begin{rem}
Notice that even if $\dot{x}^d(\tau) = F(x^d(\tau)) + g(x^d(\tau)) u^d(\tau))$, it still doesn't hold $x^d_t \neq \phi_{\tau_s}(x^d_{t-1},u^d_{t-1})$, since the continuous-time trajectories are computed without consideration of the zero-order hold actuation.
\end{rem}
\noindent We will take the following linear approximation of $\phi_{\tau_s}$ around the reference trajectory:
\begin{align}\label{eq:approx}
\phi_{\tau_s}(x_{t-1},u_{t-1}) \approx x^d_t + \underbrace{\exp(\nabla F|_{x^d_{t-1}}\tau_s)}_{=:\hat{A}_{t-1}} *(x_{t-1}-x^d_{t-1}) + \underbrace{\int^{\tau_s}_{0} \exp(\nabla F|_{x^d_{t-1}}\tau) g(x^d_{t-1}) d \tau}_{=:\hat{B}_{t-1}} * (u_{t-1}-u^d_{t-1})
\end{align}
Denote $[A_{t-1},B_{t-1}] := \nabla_{x,u}\phi_{\tau_s}|_{(x^d_{t-1},u^d_{t-1})}$ the true linearization of $\phi_{\tau_s}$ at $(x^d_{t-1},u^d_{t-1})$ and notice that the above approximation \eqref{eq:approx} is only an approximation of the linearization, since $[\hat{A}_{t-1},\hat{B}_{t-1}] \neq [A_{t-1},B_{t-1}]$. Using Taylor's theorem and assuming $\phi_{\tau_s}$ is differentiable, we can write 
$\phi_{\tau_s}$ as 
\begin{align}
\phi_{\tau_s}(x_{t-1},u_{t-1}) = \phi_{\tau_s}(x^d_{t-1},u^d_{t-1}) + A_{t-1}(x_{t-1}-x^d_{t-1})+B_{t-1}(u_{t-1}-u^d_{t-1}) + r_{t-1}(x_{t-1}-x^d_{t-1},u_{t-1}-u^d_{t-1})
\end{align}
where $\lim_{|z|\rightarrow 0} |r_{t-1}(z)|/|z| = 0$. We can factor out equation \eqref{eq:dtphi} into the following components
\begin{align}\label{eq:syseq}
x_t = x^d_t + \hat{A}_{t-1}(x_{t-1}-x^d_{t-1}) + \hat{B}_{t-1}(u_{t-1}-u^d_{t-1}) + e_t + e'_t(x_{t-1},u_{t-1}) 
\end{align}
where $e_t$ and $e'_t(x_{t-1},u_{t-1})$ are disturbance terms introduced due to errors in the reference trajectory and linearization
\begin{subequations} 
\begin{align}
e_t &= \phi_{\tau_s}(x^d_{t-1},u^d_{t-1})-x^d_t\\
 e'_t(x_{t-1},u_{t-1})  &= (A_{t-1}-\hat{A}_{t-1})(x_{t-1}-x^d_{t-1}) + (B_{t-1}-\hat{B}_{t-1})(u_{t-1}-u^d_{t-1}) + r_{t-1}(x_{t-1}-x^d_{t-1},u_{t-1}-u^d_{t-1}).
\end{align}
\end{subequations}
and the remaining terms represent our linear approximation of the dynamics.
We can express this more compactly in operator form:
Define $\b{\tilde{F}}^\phi\in \Csxy{\ell^n}{\ell^n}$ with the components
$$\tilde{F}^\phi_t(x_{t:0},u_{t:0}) := x^d_t + \hat{A}_{t-1}(x_{t-1}-x^d_{t-1}) + \hat{B}_{t-1}(u_{t-1}-u^d_{t-1}) $$
and the residual $\b{\Delta}^\phi\in\Cxy{\ell^n}{\ell^n}$ with the components
$\Delta^\phi_t(x_{t:0},u_{t:0}) = e_t + e'_{t}(x_{t-1},u_{t-1}) $, then we can factor out the original equations of \eqref{eq:dtphi} as
\begin{align}
\b{x} = \b{\tilde{F}}^\phi(\b{x},\b{u}) + \b{\Delta}^\phi(\b{x},\b{u})
\end{align}
where per definition we have the decomposition $\b{\tilde{F}}^\phi + \b{\Delta}^\phi = \b{F}^\phi$.\\

Our approach for synthesis is now to use $\b{\tilde{F}}^\phi$ as a model to the design a system level controller and treat $\b{\Delta}^\phi$ as disturbance terms we want to be robust against. We do this, by first solving for CLMs $\b{\tilde{\Psi}} = (\b{\tilde{\Psi}}^{x},\b{\tilde{\Psi}}^{u})$ of $\b{\tilde{F}}^\phi$ and then choosing our feedback controller as $\SL{\b{\tilde{\Psi}}^{x}}{\b{\tilde{\Psi}}^{u}}$.\\

Due to \thmref{thm:sufnec}, $\b{\tilde{\Psi}}$ is a CLM of $\b{\tilde{F}}^\phi$ if and only if it satisfies the CLM equation \eqref{eq:opfunceq_red} for $\b{\tilde{F}}^\phi$, i.e: $\b{\tilde{\Psi}}^{x} = \b{\tilde{F}}^\phi(\b{\tilde{\Psi}})+\b{I}$. We will restrict $\b{\tilde{\Psi}}$ to be of the affine form $\b{\tilde{\Psi}} = (\b{r},\b{m}) + (\b{R},\b{M})$ with $(\b{r},\b{m})\in \ell^n\times \ell^m$ and $\b{R} \in \LCxy{\ell^n}{\ell^n}$, $\b{M} \in \LCxy{\ell^n}{\ell^m}$. Thus the component functions of $\b{\tilde{\Psi}}$ take the from 
\begin{align} \label{eq:clm-ansatz}
& \tclm{x}{t}(\alpha_{1:t+1}) = \sum^{t+1}_{k=1} R_{t,k}\alpha_{k} + r_t\\
 &\tclm{u}{t}(\alpha_{1:t+1}) = \sum^{t+1}_{k=1} M_{t,k}\alpha_{k} + m_t
\end{align}
where $R_{t,j}\in\R^{n \times n}$, $M_{t,j} \in \R^{n \times m}$, $r_t$, $m_t$ are some fix sequences and $\alpha_{1:t+1}$ denote the $t+1$ arguments of the component function. Structuring $\b{\tilde{\Psi}}$ in this form and using linearity of $\b{\tilde{F}}^\phi$ reduces the original operator equation $\b{\tilde{\Psi}}^{x} = \b{\tilde{F}}^\phi(\b{\tilde{\Psi}})+\b{I}$ simply to the following set of linear equations for $R_{t,j}$, $M_{t,j}$ and $r_t$, $m_t$:
\begin{align}\label{eq:LTVcond-2}
&R_{t,k} = \hat{A}_{t-1}R_{t-1,k-1} + \hat{B}_{t-1}M_{t-1,k-1}\text{ for all }k\leq t, \quad R_{t,1} = I\\
&r_t = x^d_t, \quad m_t =  u^d_t
\end{align}
Equation \eqref{eq:LTVcond-2} is an affine subspace constraint and opens up many possible ways to synthesize for solutions $R_{t,k}$, $M_{t,k}$. In fact, \eqref{eq:LTVcond-2} matches the linear time-varying formulation of SLS as discussed in \cite{ho2019scalable}, \cite{anderson2019system} and for our case-study here, we are synthesizing for $R_{t,k}$, $M_{t,k}$ by solving the following $\mathcal{H}_2$/ LQR problem for the LTV system $\b{\tilde{F}}^\phi$ subject to an FIR constraint with horizon $T=60$ time-steps:
\begin{align}\label{eq:QP}
\begin{array}{rc}\min\limits_{R_{t,k},M_{t,k}} &\sum\limits_{0 \leq t \leq H}\sum\limits_{{ 1 \leq k \leq T}}\|R_{t,k}\|^2_{F}+\|M_{t,k}\|^2_{F}\\
s.t.&R_{t,k} = \hat{A}_{t-1}R_{t-1,k-1} + \hat{B}_{t-1}M_{t-1,k-1}\\
 & R_{t,1} = I, \quad R_{t,T} = 0
\end{array}
\end{align}
 \begin{rem}
 See \cite{anderson2019system} for details of the $\mathcal{H}_2$/LQR problem setup and derivation of the convex optimization problem.
The FIR horizon can be understood as a time-window $[t,t+T]$ given to the controller to kill off the disturbance $\hat{w}_t$. Considering the sampling time of our example, $T=60$ translates here to a $2$-second window in continuous-time. 
\end{rem}
Furthermore, $H$ denotes the length of the trajectory in sampling time-steps. The above problem can be solved in closed form, since it is a QP without inequality constraints. Moreover, the change of variables $R_{j+h,j+1}$,$M_{j+h,j+1}$ with $0\leq j\leq T-1$, $0\leq h \leq H$ shows that \eqref{eq:QP} can be decomposed over $h$ into $H$ separate QP's that can be solved analytically and in parallel, hence showing that the computational complexity of our synthesis approach is \textbf{independent} of the trajectory length $H$.

The solutions of \eqref{eq:QP} are taken to parametrize the operators \eqref{eq:clm-ansatz} which give us the system level controller $\SL{\b{\tilde{\Psi}}^{x}}{\b{\tilde{\Psi}}^{u}}$. The resulting closed loop of the cart pole system \eqref{eq:dtphi} and controller $\SL{\b{\tilde{\Psi}}^{x}}{\b{\tilde{\Psi}}^{u}}$ can be put into the form of our robust stability analysis in \secref{sec:robstab}:
\begin{subequations}
\begin{align}
x_{t} &= \phi_{t_s}(x_{t-1},u_{t-1}) + w_t\\
\hat{w}_t &= x_t+v_t-x^d(t\tau_s)-\sum^{t+1}_{k=2} R_{t,k}\hat{w}_{t+1-k} \\
u_t & = u^d(t\tau_s)+\sum^{t+1}_{k=1} M_{t,k}\hat{w}_{t+1-k}+d_t
\end{align}
\end{subequations}
Furthermore, referring to \thmref{thm:redstab}, it can be verified that the residual operator $\b{\Delta}^\phi$ we defined earlier matches the residual operator of \thmref{thm:redstab}, i.e.: $\b{\Delta}[\b{\tilde{F}},\b{\tilde{\Psi}}] = \b{\Delta}^\phi$. Thus, \thmref{thm:redstab} applies to our problem setting directly. More specifically, the local result \lemref{lem:smallgain} can be used to obtain robust stability guarantees: If the lumped residual terms are $(\gamma,\beta)$ $\ell_p$-stable with $\gamma<1$, then the closed loop system is f.g. $\ell_p$-stable for small enough perturbations.

 \section{ Blending SL controllers}\label{sec:blend-app}
\noindent We will present a class of SL controllers that can be viewed as nonlinear "blends" of multiple linear controllers and demonstrate their use in application to linear systems with state and input constraints. Consider the maps $\CLM{} = (\CLM{x},\CLM{u})$ of the form:
\begin{subequations}
\label{eq:NLblend}
\begin{align}
\CLM{x} &= \b{I} + \sum^{N}_{i=1} (\b{R}^{i}-\b{I})\b{G}^{i} \\
 \CLM{u} &= \sum^{N}_{i=1} \b{M}^{i}\b{G}^{i} 
\end{align}
 \end{subequations}
where $\b{R}^i\in \LCxy{\ell^n}{\ell^n}$, $\b{M}^i\in \LCxy{\ell^n}{\ell^m}$ are $N$ linear operators. In addition, $\b{R}^{i}$ are chosen such that $\b{R}^{i}-\b{I}\in\LCsxy{\ell^n}{\ell^n}$. Let $R^i_{t,k} \in \R^{n \times n}$ and $M^i_{t,k} \in \R^{n \times m}$ be the matrices associated with the component functions $R^i_t$, $M^i_t$ of $\b{R}^i$, $\b{M}^i$:
\begin{align}
&R^i_t(z_{1:t+1}) = \sum^{t+1}_{k=1} R^i_{t,k}z_k &M^i_t(z_{1:t+1}) = \sum^{t+1}_{k=1} M^i_{t,k}z_k.
\end{align}
\begin{rem}
 $\b{R}^{i}-\b{I}$ being strictly causal is equivalent to $R^i_{t,1} = I$ for all $t$ and $i$.
 \end{rem}
\noindent Using the above notation, the implementation of $\SL{\CLM{x}}{\CLM{u}}$ can be written out as
\begin{subequations}
\label{eq:nl-blend}
\begin{align}
\label{eq:nl-blend-1}
    \hat{w}_t &= x_t - \sum^{N}_{i=1}\sum^{t+1}_{k=2}R^{i}_{t,k}\tilde{w}^i_{t+1-k}\\
    \tilde{w}^i_t &= G^i_t(\hat{w}_{t:0})\\
    \label{eq:nl-blend-2} u_t &= \sum^{N}_{i=1}\sum^{t+1}_{k=1}M^{i}_{t,k}\tilde{w}^i_{t+1-k}
\end{align}
\end{subequations}
\begin{rem}
In the above implementation we are assuming explicit knowledge of the functions $G^i_t(.)$ and therefore $\tilde{w}$ is not an internal state, but merely a placeholder. In an implementation where $G^i_t(.)$ were to be computed implicitly via recursion, additional internal states would need to be added and included in the stability analysis. 
\end{rem}

\noindent Next, we will derive under which conditions the operator $\CLM{}$ of the form \eqref{eq:NLblend} is a closed loop map of a general linear causal system. Define a system in operator form
\begin{align}\label{eq:linF}
\b{x} = \b{F}(\b{x},\b{u}) +\b{w}
\end{align}
where we assume $\b{F}:\ell^n\times \ell^m \mapsto \ell^n$ to be a linear, strictly causal operator. Due to linearity, we can split $\b{F}$ into two operators $\b{F}^x:\ell^n \mapsto \ell^n$, $\b{F}^u:\ell^m \mapsto \ell^n$ and write $\b{F}$ as $\b{F}(\b{x},\b{u}):=(\b{F}^x(\b{x})+\b{F}^u(\b{u}))$. Correspondingly, $\b{F}$ has the component functions
\begin{align}
F_t(x_{t:0},u_{t:0}) = \sum^{t+1}_{k=1}F^x_{t,k} x_{t+1-k} + \sum^{t+1}_{k=1}F^u_{t,k} u_{t+1-k}
\end{align} 
where $F^x_{t,k} \in \R^{n \times n}$, $F^u_{t,k} \in \R^{n \times m}$ and due to the strict causality we have 
\begin{align}
F^x_{t,1} = 0,\quad  F^u_{t,1} = 0, \quad \text{for all }t.
\end{align}
With the above definitions the usual system description is
\begin{align}
x_{t} = \sum^{t}_{k=1}F^x_{t,k+1} x_{t-k} + \sum^{t}_{k=1}F^u_{t,k+1} u_{t-k} + w_t
\end{align}
Now recall that per construction, the implementation $\SL{\CLM{x}}{\CLM{u}}$ ensures the identity $\CLM{x}\b{\hat{w}} = \b{x}$, $\CLM{u}\b{\hat{w}} = \b{u}$. Plugging this relation into \eqref{eq:linF} and using \eqref{eq:nl-blend} to simplify the equations results in the following dynamics for $\hat{w}$
\begin{align}\label{eq:intw-ref}
\b{\hat{w}} = \b{F}^x(\b{I}-\sum^{N}_{i=1} \b{G}^i)\b{\hat{w}}  + \sum^{N}_{i=1} \b{\Delta}^i \b{G}\b{\hat{w}}  + \b{w}
\end{align}
In light of our discussion in \secref{sec:robstab}, the solutions $\b{\hat{w}}$ of \eqref{eq:intw-ref} characterize the partial map $\b{\Phi}^{\hat{w}}_{S'_1}[\b{F},\CLM{}]$ and it holds $\b{\hat{w}}=\b{\Phi}^{\hat{w}}_{S'_1}[\b{F},\CLM{}](\b{w},\b{0},\b{0})$.
We denoted $\b{\Delta}^i$ as the residual linear operators
\begin{align}\label{eq:resid_i}
\b{\Delta}^i:= \b{F}^x\b{R}^i+\b{F}^u\b{M}^i + \b{I} - \b{R}^i.
\end{align}
which are capturing the residual terms of the CLM equation of the maps $(\b{R}^i,\b{M}^i)$. From the above equation, we see that if we choose each of $(\b{R}^i,\b{M}^i)$ to be feasible closed-loop maps for the system \eqref{eq:linF} (i.e. $\b{\Delta}^i = 0$) and also choose $\b{G}^i$ to satisfy $\sum^{N}_{i=1} \b{G}^i = \b{I}$, then the operator $\CLM{} = (\CLM{x},\CLM{u})$ of \eqref{eq:nl-blend} becomes a CLM of the linear system \eqref{eq:linF}. This is summarized in the Lemma below:
\begin{lem}\label{lem:nlblend}
Assume $\b{F}$ is linear and strictly causal and define $\CLM{}=(\CLM{x},\CLM{u})$ according to \eqref{eq:nl-blend} with some linear $(\b{R}^i,\b{M}^i) \in \feas[\b{F}]$ and potentially nonlinear $\b{G}^i$. Then, if $\b{G}^i$ satisfy $\sum^{N}_{i=1} \b{G}^i = \b{I}$,
then $\CLM{}\in \feas[\b{F}]$.
\end{lem}
\noindent We will discuss the importance of the above nonlinear system level controllers for control design problems in linear systems with input saturation and state constraints.
\subsection{Linear Systems with Input Saturation and State Constraints}

\noindent Let's consider the LTI system $H$
\begin{align}\label{eq:lin-0}
H:\quad x_{t} &= Ax_{t-1} + Bu_{t-1} + w_t
\end{align}
and the corresponding nonlinear system $H'$ obtained by modifying \eqref{eq:lin-0} to have actuator saturation:
\begin{align}\label{eq:linsat}
H':\quad x_{t} &= Ax_{t-1} + B\sat(u_{t-1}|\cl{U}) + w_t.
\end{align}
We will define $\sat$ as a generalized notion of a saturation function: Given some closed bounded convex set $\mathcal{U}$ with $0\in \cl{U}$, we will take $\sat (\cdot | \cl{U}) : \R^m \mapsto \cl{U}$ to be the projection map onto the set $\cl{U}$ with the following property:
\begin{align}
\begin{array}{rl} |\sat(u| \cl{U})-u| = \min & |u'| \\
\mathrm{s.t.} & u + u' \in \cl{U}
\end{array}
\end{align}
\begin{rem}
We would like to point out that the following results are formulated w.r.t to some general norm $|.|$ in $\R^n$. Results w.r.t to a particular norm $|.|_p= (\sum^{n}_{i=1} |x_i|^p)^{1/p}$ can be obtained, by replacing $|.|$ with $|.|_p$ in all statements. Moreover, recall the corresponding implied definitions of $\ell^n_p$ and $\|\cdot\|_{p}$ from \secref{sec:prelim}.
\end{rem}
\noindent Correspondingly, define $\b{F}^{H}\in\LCsxy{\ell^n\times \ell^m}{\ell^n}$ and $\b{F}^{H'}\in\Csxy{\ell^n\times \ell^m}{\ell^n}$ with the component functions:
\begin{align}
&F^{H}_t(x_{t:0},u_{t:0}) =Ax_{t-1} + Bu_{t-1} &F^{H'}_t(x_{t:0},u_{t:0}) =Ax_{t-1} + B\sat(u_{t-1}|\cl{U})
\end{align} 

\subsubsection{Satisfying State Constraints in Unsaturated Regime}
Assume that given some specified convex sets $\cl{X}\in \R^n$ and $\cl{W}\in\R^n$, we want to design a controller for the nonlinear system $H'$ such that for any disturbance sequence $\b{w}$ in $\cl{W}$, (i.e. $w_t \in \cl{W}$ for all $t$) the state $x_t$ is always guaranteed to stay within the set $\cl{X}$ if $x_0 \in \cl{X}$. A general approach to tackle this problem within the context of SLS has been presented in \cite{chen2019system}. For polytopic sets $\cl{X}$, $\cl{U}$ and $\cl{W}$, the authors in \cite{chen2019system} propose an efficient method to synthesize controllers $\SL{\b{R}}{\b{M}}$ for this problem. By casting the problem as a convex optimization problem, \cite{chen2019system} computes linear-time invariant CLMs $(\b{R},\b{M}) \in \feas(\b{F}^H)$ of the linear system $H'$ that satisfy:
\begin{align}\label{eq:invarsls}
\forall \b{w}\in \cl{W}:\quad \b{R}(\b{w}) \subset \cl{X},\quad \b{M}(\b{w}) \subset \cl{U}
\end{align}
\begin{rem}
For sets $\cl{S} \subset \R^n$, we will define the notation $\b{a} \subset \cl{S}$ to mean $a_t \in \cl{S}$ for all $t$.
\end{rem}
\noindent It is clear from this setup that for sequences $\b{w} \in \cl{W}$, the map $(\b{R},\b{M})$ is a CLM for the linear system $H$ as well as the nonlinear system $H'$, (i.e.: $(\b{R},\b{M}) \in \feas(\b{F}^{H'})$). This is because $\CLM{u}$ is designed to never actually saturate the actuator for disturbances in $\cl{W}$. Hence, the resulting controller $\SL{\b{R}}{\b{M}}$ solves to the original problem for the nonlinear system $H'$. Recalling our definition of the closed loop mappings with system level implementations from \secref{sec:robstab}, this can be expressed as: $\b{\Phi}_{S'_1}[\b{F}^{H'},(\b{R},\b{M})] (\b{w}) \subset \cl{X} \times \cl{U} $ \\

\subsection{Stability and Convergence to Target Set in Saturated Regime }
\noindent Since modeling uncertainties are unavoidable when dealing with the real world, it is hard to have complete certainty whether a disturbance $w_t$ will satisfy $w_t\in\cl{W}$ for all time in a practical application. Therefore, a natural extension of the above problem setting is that we require the controller to degrade performance gracefully in case the disturbance leaves the set $\cl{W}$ occasionally and as a consequence actuator saturation does occur. It is commonly known that graceful performance degradation can not be taken for granted, as instability phenomenon like the "wind-up" effect can occur if controller synthesis improperly deals with actuator saturation. For stable system matrices $A$, \cite{chen2019system} shows a modification based on the IMC principle of the controller $\SL{\b{R}}{\b{M}}$ that guarantees closed loop stability in the case of $\b{w} \notin \cl{W}$. Next, we will propose a simple modification based on the nonlinear SLS approach that comes with additional of results for stability and performance analysis of the nonlinear closed loop $H'$:\\

\noindent We will augment the previous linear controller $\SL{\b{R}}{\b{M}}$ by incorporating the maps $\b{R}$, $\b{M}$ into a nonlinear system level controller, such that $\ell_p$ stability of the nonlinear closed loop of $H'$ is guaranteed. Moreover, global (for stable $A$) and local (for unstable $A$) stability results and corresponding transient bounds are derived for the closed loop. Convergence to $\cl{X}$ in finite time is shown for $\ell_p$ perturbations with ($p<\infty$). 

\noindent Take $\b{R}$, $\b{M}$ to be solutions that satisfy \eqref{eq:invarsls} from the previous sections and define the nonlinear maps $\CLM{} = (\CLM{x},\CLM{u})$ as 
\begin{subequations}
\label{eq:NLblend-sat-0}
\begin{align}
\CLM{x} &= \b{I} + (\b{R}-\b{I})\b{G} + (\b{R}^{'}-\b{I})\b{G}^{'} \\
 \CLM{u} &= \b{M}\b{G} + \b{M}^{'}\b{G}^{'}
\end{align}
 \end{subequations}
 where we choose $\b{G}$ and $\b{G}'$ as 
 \begin{align}
 G_t(w_{t:0}) &= \sat (w_t | \cl{W}) \\
 G'_t(w_{t:0}) &= w_t-\sat (w_t | \cl{W})
 \end{align}
and we choose the matrices of $\b{R}'$ and $\b{M}'$ as
\begin{align}
 &R'_{t,k} = \left\{ \begin{array}{cl} A^{k-1},&\forall k\leq \bar{T}+1 \\ 0 &\text{else}  \end{array}\right. &  &M'_{t,k} = 0 
 \end{align}
 
Now, notice that because of our choice of $\b{G}$ and $\b{G}'$ we trivially satisfy $\b{G}+\b{G}' = \b{I}$. Furthermore, since $\b{G}$ maps onto $\cl{W}$ and $\b{M}(\b{w}) \in \cl{U}$ for any $\b{w} \in \cl{W}$, we immediately have $\sat(\b{M} \b{G}|\cl{U}) = \b{M} \b{G}$. Recalling that per design we chose $\b{M}'=0$, we can conclude that:

The closed loop of $\SL{\CLM{x}}{\CLM{u}}$ with $H'$ can be written as the new system $H''$:
\begin{subequations}
\label{eq:nl-cl-2}
\begin{align}
H'':\quad 	x_t &= A x_{t-1} + B\sat (u_{t-1} | \cl{U}) + w_t\\
    \hat{w}_t &= x_t - \sum^{t+1}_{k=2}R_{t,k}(\sat (\hat{w}_{t+1-k} | \cl{W})) - \sum^{\min\{t+1,\bar{T}+1\}}_{k=2}R'_{t,k}(\hat{w}_{t+1-k} -\sat (\hat{w}_{t+1-k} | \cl{W}))\\    
	 u_t &= \sum^{t+1}_{k=1}M_{t,k}(\sat (\hat{w}_{t+1-k} | \cl{W}))
\end{align}
\end{subequations}
and denote $\b{\Phi}_{H''}:\b{w} \mapsto (\b{x},\b{u},\b{\hat{w}})$ the corresponding closed loop map.

\noindent As discussed in \secref{sec:robstab}, checking stability of the partial map $\b{\Phi}^{\hat{w}}_{H''}$ is sufficient to guarantee internal stability of the closed loop (i.e. stability in presence of general perturbations $\bm{\delta}$ as introduced in closed loop $S'_1$ in \secref{sec:robstab}). 
 We will follow the derivation of the dynamics $\b{\hat{w}}=\b{\Phi}^{\hat{w}}_{H''}(\b{w})$ as in \secref{sec:robstab}.
 Now, notice that since $\b{M}'=0$ and the $\b{G}(\b{\hat{w}}) = \sat (\hat{w}_t| \cl{W}) \in \cl{W}$ the saturation function satisfies
 \begin{align}
\notag \sat(\b{M} \b{G}+\b{M}' \b{G}'| \cl{U}) &= \sat(\b{M} \b{G}| \cl{U})\\
  &= \b{M} \b{G} +\b{M}' \b{G}'
 \end{align}
  and the residual $\b{\Delta}[\b{F}^{H'},\CLM{}] = \b{F}^{H'}(\CLM{}) + \b{I} - \b{\CLM{x}}$ can be split into the terms $\b{\Delta}$ and $\b{\Delta}'$ that are the residual operators corresponding to $(\b{R},\b{M})$ and $(\b{R}',\b{M}')$ w.r.t. to the linear system $H$:
\begin{align}\label{eq:resid_i}
&\b{\Delta}:= \b{F}^{H}(\b{R},\b{M}) + \b{I} - \b{R} & \b{\Delta}':= \b{F}^{H}(\b{R}',\b{M}') + \b{I} - \b{R}'
\end{align} 
The closed loop dynamics $\b{\hat{w}}=\b{\Phi}^{\hat{w}}_{H''}(\b{w})$ then take the form of equation \eqref{eq:intw-ref}:
\begin{align}\label{eq:intw-ref-2}
\b{\hat{w}} =   \b{\Delta} \b{G}\b{\hat{w}} + \b{\Delta}' \b{G}'\b{\hat{w}}  + \b{w}
\end{align}
Since $(\b{R},\b{M})$ are CLMs of the linear system $H$, we have $\b{\Delta} =0$. Moreover, since $\b{M}'=0$, then also $\b{\Delta}'$ reduces to 
 \begin{align}
\Delta'_t(a_{t:0}) := \left\{ \begin{array}{cl} A^{\bar{T}}a_{t-\bar{T}},&\forall t\geq \bar{T} \\ 0 &\text{else}  \end{array}\right.
\end{align} 
and the dynamics of $\b{\hat{w}}=\b{\Phi}^{\hat{w}}_{H''}(\b{w})$ reduce to
\begin{align}\label{eq:int-dyn}
\hat{w}_t = A^{\bar{T}}(\hat{w}_{t-\bar{T}}-\sat(\hat{w}_{t-\bar{T}}|\cl{W})) + w_t
\end{align}
Define $\cl{B}_\eta := \{w|\:\: |w| < \eta \}$ to be the ball of radius $\eta$ corresponding to the norm $|\cdot |$.
The following closed loop stability result follows:
\begin{lem}\label{eq:awp} 
Define $\bar{\eta}:= \sup\{\eta|\:\:\cl{B}_\eta \subset \cl{W}  \}$ and assume $\bar{\eta}>0$. Then, it holds:
\begin{enumerate} 
\item For any $0\leq\gamma<\min\{1,|A^{\bar{T}}|\}$ holds: $$\|\b{w}\|_p \leq (1-\gamma)\frac{|A^{\bar{T}}|\bar{\eta}}{|A^{\bar{T}}|-\gamma} \quad \implies \quad \|\b{\hat{w}}\|_p \leq \frac{1}{1-\gamma}\|\b{w}\|_p $$
\item If $|A^{\bar{T}}| <1$, then the following bound holds for all $\b{w}$: $$ \|\b{\hat{w}}\|_p \leq \frac{1}{1-|A^{\bar{T}}|}\|\b{w}\|_p$$
\end{enumerate}

\end{lem} 
\begin{proof}
Then the following relationship can be easily derived from the generalized definition of the saturation function:
\begin{align}
\begin{array}{rl} |\sat(w| \cl{W})-w| = \min & |w'|  \\
\mathrm{s.t.} & w + w' \in \cl{W}
\end{array}
\begin{array}{rl} \leq \min & |w'| \\
\mathrm{s.t.} & |w + w'| < \bar{\eta}
\end{array}
\begin{array}{rl} \leq \min_t & t|w| \\
\mathrm{s.t.} & (1-t)|w| < \bar{\eta}
\end{array}
\begin{array}{l} = \max\{0,|w|-\bar{\eta}\}\\
\phantom{0}
\end{array}
\end{align}
From the above inequality, it becomes clear that $|\sat(w| \cl{W})-w| < |w|$ for all $w \in \R^n$. Now, if $|A^{\bar{T}}| <1$ we can conclude that the residual operator $\b{\Delta}''$, 
\begin{align}
\Delta''_t(\hat{w}_{t:0}) := \left\{ \begin{array}{cl} A^{\bar{T}}(\hat{w}_{t-\bar{T}}-\sat(\hat{w}_{t-\bar{T}}|\cl{W})),&\forall t\geq \bar{T} \\ 0 &\text{else}  \end{array}\right.
\end{align}
satisfies $\Nxy{\b{\Delta}''(\b{a})}{p} \leq |A^{\bar{T}}| \|\b{a}\|_p < \|\b{a}\|_p$. On the other hand, for any $\gamma<\min\{1,|A^{\bar{T}}|\}$, it can be verified that the following local small-gain property holds:
\begin{align}
\|\b{a}\|_p < \frac{|A^{\bar{T}}|\bar{\eta}}{|A^{\bar{T}}|-\gamma} \quad \implies \quad \Nxy{\b{\Delta}'(\b{a})}{p} \leq \gamma \|\b{a}\|_p
\end{align}
 The desired local and global results follow by direct application of the small gain result \lemref{lem:smallgain}.
\end{proof}
Recall the following fact from linear algebra: 
\begin{lem}[Gelfand's Theorem]
Denote $\rho(A)$ as the spectral radius (max absolute value of eigenvalue) of $A \in \R^{n\times n}$, then for any matrix norm $|.|$ holds $\lim_{k \rightarrow \infty} (|A^k|)^{1/k} = \rho(A)$.
\end{lem}
Due to Gelfand's theorem, existence of $T$ such that $|A^T|<1$ is guaranteed if $A$ is schur, i.e. $\rho(A)<1$. This is true regardless of which particular induced-norm $|\cdot|$ we choose. Notice that \eqref{eq:awp} also gives local stability even if $A$ is unstable.

As a corollary of the above result, we can also conclude that there exists a time $t'$ for which $x_t$ is guaranteed to stay in $\cl{X}$ for all time $t>t'$ despite any $\ell_p$ perturbations (with $p<\infty$):

\begin{coro}
Assume $\bar{\eta}>0$ with the definition used in \lemref{eq:awp}. Now, if $\b{\hat{w}}\in \ell^n_p$, then there exists a time $t'$ such that for all $t>t'$: $x_t\in\cl{X}$.
\end{coro}
\begin{proof}
Recall the relationship $\b{\hat{w}} = (\CLM{x})^{-1}\b{x} \Leftrightarrow \CLM{x}\b{\hat{w}} = \b{x}$. Plugging in \eqref{eq:NLblend-sat-0} and performing some simplification allows us to decompose $\b{x}$ into the terms $\b{s}$ and $\b{s}'$:
\begin{align}
\b{x} = \underbrace{\b{R}(\sat(\b{\hat{w}}|\cl{W}))}_{=:\b{s}} + \underbrace{\b{R}'(\b{\hat{w}}-\sat(\b{\hat{w}}|\cl{W}))}_{=:\b{s}'} 
\end{align}
Now notice that since $\sat(\b{\hat{w}}|\cl{W}) \in \cl{W}$ per definition and the design of $\b{R}$, we know that $\b{s} \in \cl{X}$. Furthermore, $\b{\hat{w}}\in \ell^n_p$ implies that $\hat{w}_t \rightarrow 0$ and since $\bar{\eta}>0$, there exists a time $\bar{t}$ such that for all $t>\bar{t}$ holds $\hat{w}_t- \sat(\hat{w}_t|\cl{W}) = 0$. Now, recall that $R'_{t,k} = 0$ for all $k > \bar{T}+1$. It immediately follows that for all $t> \bar{t}+\bar{T}+1$, $s'_t$ is zero. Thus, for all $t> \bar{t}+\bar{T}+1$ holds $x_t = s_t \in \cl{X}$.
\end{proof}

\end{document}